\documentclass{article}
\usepackage{fullpage}
\usepackage[ruled,linesnumbered,vlined]{algorithm2e}

\usepackage{xcolor}
\usepackage{dirtytalk}
\usepackage{amsfonts}
\usepackage{amsthm}
\usepackage{amsmath}
\usepackage{amssymb}
\usepackage{mathtools}
\usepackage{bbm}
\usepackage[hidelinks]{hyperref}
\usepackage{cleveref}
\usepackage{enumerate}
\usepackage{comment}
\usepackage[noend]{algpseudocode}

\usepackage[]{natbib}
\usepackage{thmtools}
\usepackage{thm-restate}
\renewcommand{\cite}[1]{\citep{#1}}
\newcommand{\textcite}[1]{\citet{#1}}

\newtheorem{theorem}{Theorem}

\newtheorem{lemma}[theorem]{Lemma}

\newtheorem{corollary}[theorem]{Corollary}

\theoremstyle{definition}

\newcommand{\remove}[1]{}

\newcommand{\sw}{\mathsf{SW}}

\newcommand{\R}{\mathbb{R}}
\newcommand{\F}{\mathcal{F}}

\newcommand{\Rplus}{\mathbb{R}_{\geq 0}}

\newcommand{\one}{\mathbf{1}}%

\DeclareMathOperator*{\argmax}{arg\,max}

\DeclareMathAlphabet{\mymathbb}{U}{BOONDOX-ds}{m}{n}

\newcommand{\welfare}{\mathsf{Welfare}}
\newcommand{\rev}{\mathsf{Rev}}
\newcommand{\uti}{\mathsf{Util}}
\newcommand{\alg}{\mathsf{Alg}}
\newcommand{\compratio}{\mathsf{CompRatio}}
\newcommand{\E}{\mathbb{E}}
\newcommand{\D}{\mathcal{D}}
\newcommand{\pr}{\mathbf{Pr}}

\newcommand{\Z}{\mathbb{Z}}
\newcommand{\allocs}{\mathbf{x}}
\newcommand{\alloci}[1][i]{x_{#1}}
\newcommand{\talloci}[1][i]{\tilde{x}_{#1}}
\newcommand{\yalloci}[1][i]{y_{#1}}
\newcommand{\ysalloci}[1][i]{y^*_{#1}}
\newcommand{\vals}{\mathbf{v}}

\newcommand{\hopt}{\operatorname{HOpt}}
\newcommand{\prophet}{\operatorname{Prophet}}
\newcommand{\onl}{\operatorname{Online}}
\newcommand{\eaopt}{\operatorname{ExAnteOpt}}
\newcommand{\dyn}{\operatorname{DynIP}}
\newcommand{\stat}{\operatorname{StatIP}}
\newcommand{\copyip}{\operatorname{SuppIP}}

\newcommand{\Pois}{\mathrm{Pois}}
\newcommand{\Binom}{\mathrm{Binom}}
\newcommand{\Unif}{\mathrm{Unif}}
\newcommand{\si}{\mathrm{SI}}

\newcommand{\countSj}{\#S_j}

\newcommand*\samethanks[1][\value{footnote}]{\footnotemark[#1]}

\title{Multi-Unit Combinatorial Prophet Inequalities}
\author{
Shuchi Chawla\thanks{The University of Texas at Austin. These authors were supported in part by NSF award CCF-2225259.} \\ {\tt shuchi@cs.utexas.edu} \and 
Trung Dang\samethanks \\ {\tt dddtrung@cs.utexas.edu}
\and
Zhiyi Huang\samethanks \\ {\tt zhiyih@cs.utexas.edu}
\and
Yifan Wang\thanks{Georgia Institute of Technology.}  \\ {\tt ywang3782@gatech.edu}
}

\date{}

\begin{document}

\allowdisplaybreaks

\maketitle

\thispagestyle{empty}
\addtocounter{page}{-1}

\begin{abstract}
We consider a combinatorial auction setting where buyers have fractionally subadditive (XOS) valuations over the items and the seller's objective is to maximize the social welfare. A prophet inequality in this setting bounds the competitive ratio of sequential allocation (often using item pricing) against the hindsight optimum. We study the dependence of the competitive ratio on the number of copies, $k$, of each item. 

We show that the multi-unit combinatorial setting is strictly harder than its single-item counterpart in that there is a gap between the competitive ratios achieved by {\em static item pricings} in the two settings. However, if the seller is allowed to change item prices {\em dynamically}, it becomes possible to asymptotically match the competitive ratio of a single-item static pricing. 
We also develop a new {\em non-adaptive anonymous} multi-unit combinatorial prophet inequality where the item prices are determined up front but increase as the item supply decreases. Setting the item prices in our prophet inequality requires minimal information about the buyers' value distributions -- merely (an estimate of) the expected social welfare accrued by each item in the hindsight optimal solution suffices. Our non-adaptive pricing achieves a competitive ratio that increases strictly as a function of the item supply $k$.

\end{abstract}

\newpage

\section{Introduction}

In recent years strong connections have emerged between Bayesian mechanism design and optimal stopping problems called prophet inequalities. One surprising implication of these connections is that for many combinatorial auction settings, social welfare can be approximately maximized through a sequential pricing mechanism. Consider, for example, a book store hosting a book fair with many titles on sale. Suppose that every buyer has unit demand with different values for different books. The store's goal is to maximize the total value buyers receive from purchasing the books. Instead of holding an auction where all buyers participate simultaneously, the store can simply place a fixed price on every item and allow buyers to purchase their favorite books at the posted prices asynchronously until supplies last. A seminal result of \citet{feldman2014combinatorial} shows that this sequential pricing mechanism obtains at worst half of the optimal expected social welfare. This type of result is called a prophet inequality.

Sequential posted pricing mechanisms such as the one described above have many nice properties that make them practical. From the viewpoint of the buyers, they are easy to understand, involve no strategizing, and do not require buyers to divulge their private values. From the viewpoint of the seller, they are easy to implement and are robust to small changes in the market (e.g. the order in which buyers arrive). The question of how well they can approximate the optimal social welfare is therefore well motivated.

In this paper we study {\bf posted price mechanisms in  combinatorial multi-unit settings}. Consider again the book fair example described above. Suppose that the seller has at least $k$ copies of each book available for sale for some $k>1$. Is it possible for the seller to price the books in a manner that obtains {\em strictly more than half} of the optimal social welfare?


This question is well-understood in the {\em single-item} setting where the seller has $k>1$ copies of just one item to sell. \citet{hajiaghayi2007automated} showed that selling the item at a fixed price until supply runs out achieves competitive ratio that increases with $k$ as $1-O(\sqrt{\log k/k})$. \citet{alaei2014bayesian} showed that even better asymptotic performance ($1-O(1/\sqrt{k})$) can be obtained by changing prices dynamically. Most recently, \citet{chawla2024static} and \citet{jiang2022tightness} obtained upper and lower bounds on the performance of different kinds of posted pricing mechanisms that are tight for {\em every value of $k$}. 

In comparison, the combinatorial setting is less well understood. We focus in this paper on settings where buyers have fractionally subadditive (XOS) values over the items. For this setting, \citet{feldman2014combinatorial} showed that item pricing obtains a competitive ratio of $1/2$, and this ratio is tight for $k=1$ even in the single item case. For $k>1$, however, it was not known prior to our work whether one could achieve a competitive ratio better that $1/2$ or potentially match the performance of \citeauthor{alaei2014bayesian} or \citeauthor{chawla2024static}'s single item pricings. 

\paragraph{Dynamic versus static pricings.} The answer to this question depends on how much flexibility the seller has in setting the prices. We distinguish, in particular, between {\em dynamic} and {\em static} item pricing mechanisms. In the former, the seller can choose prices for a buyer based on the entire history of the mechanism up to that point including the identities of previously arrived buyers, their valuations, and the item supply left. In a static item pricing, in contrast, the seller sets prices on items up front, and the same prices are offered to each buyer as long as supplies last. Dynamic item pricings offer more flexibility to the seller and can therefore potentially obtain a better competitive ratio. One may further ask whether general online allocation algorithms, that do not offer the buyers a pricing at all and may not even be truthful, can obtain an even better competitive ratio.

Surprisingly, in the combinatorial single-unit ($k=1$) setting of \citeauthor{feldman2014combinatorial} with XOS buyers, there is no gap between general online allocation and static item pricing: the latter obtains a competitive ratio of $1/2$, which is tight for online algorithms even in the setting of allocating a single item to two buyers. 
On the other hand, for the single item $k$-unit setting, \citet{jiang2022tightness} prove that static item pricing is strictly weaker than dynamic item pricing.\footnote{For the single item setting, every online allocation algorithm can be trivially seen as a dynamic pricing algorithm.} 

We additionally consider a third kind of sequential item pricing mechanism that is stronger than static item pricing but shares many of its nice properties. In {\em supply-based static pricing} the price of each item can depend upon the amount of supply of the item left, but the same prices are offered to each buyer regardless of their identity or the timing of their arrival. For example, the seller at the book fair may price the first 10 copies of a book at a discounted price and the remaining copies at a higher amount. Observe that supply-based pricing is anonymous like static item pricing, and does not require buyers to divulge their identities or value functions to the mechanism. 

\paragraph{The role of information.} A nice feature of \citeauthor{feldman2014combinatorial}'s $1/2$-competitive pricing is that the mechanism requires very little information about the instance to compute the prices. In particular, the price of every item is set equal to half the total welfare accrued by that item in the optimal solution. Observe that this welfare contribution of each item depends only on the values obtained by the buyers for the items they receive and not on their entire valuation function (other than its influence on the optimal allocation itself). The mechanism designer can therefore estimate these welfare contributions for a given target allocation far more easily than estimating the entire value distribution of each buyer. We investigate whether improved competitive ratios as a function of the item supply $k$ can be obtained using this limited information.


\vspace{0.1in}
\noindent
We obtain the following results.
\begin{enumerate}
    \item We show that the $k$-unit combinatorial setting is strictly harder than the $k$-unit single-item setting even when buyers have unit demand, in that the competitive ratio of static item pricing in the former setting is strictly smaller than that in the latter setting. (Section~\ref{sec:hard-instance})  
    \item We develop a supply-based item pricing algorithm for the $k$-unit setting with XOS buyers that achieves a competitive ratio of $1-(k/(k+1))^k$ for all $k\ge 1$. As $k\rightarrow\infty$, this ratio converges to $1-1/e$. The prices in this mechanism depend only on the welfare contribution of each item in the optimal allocation. (Section~\ref{sec:pricing}) 
    \item We show that for the $k$-unit setting with XOS buyers dynamic item pricing can achieve a competitive ratio that tends to $1$ as $k\rightarrow\infty$ at the rate of $1-O(\sqrt{\log k/k})$. (Section~\ref{sec:dynamic})
    \item Finally, we show that for the $k$-unit setting with XOS buyers, general online allocation mechanisms can exactly match the competitive ratio of a single-item $k$-unit dynamic pricing, namely $1-1/\sqrt{k+3}$ for all $k\ge 1$. (\Cref{sec:online})
    \item Our positive results extend seamlessly to buyers with multi-unit demand. (Appendix~\ref{sec:multi_unit})
\end{enumerate}
We emphasize that in all of the results stated above, the number of distinct items for sale, $m$, and the number of buyers, $n$, can be arbitrarily large (and indeed much larger than $k$). A key feature of our positive results is obtaining competitive ratios independent of $m$.



\subsection{Technical challenges and contributions}

We now discuss our main technical contributions in more detail. In the single item setting, the key to obtaining a competitive ratio that grows with the item supply $k$ is to exploit concentration in demand. Indeed as $k$ grows, if we target setting a price that in expectation sells about $k-\sqrt{k\log k}$ copies of the item, with high probability we do not run out of the item supply. Consequently any buyers with very high values are nearly guaranteed to be served no matter when they arrive. A similar logic can be applied in some combinatorial settings (e.g. with unit demand buyers) if the number of items is small. \citet{chawla2017stability} noted that the dual prices corresponding to the natural (ex-ante) LP relaxation of social welfare maximization with scaled down supply can obtain an expected ratio of $1-\sqrt{\log m/k}$ where $m$ is the number of distinct items. Essentially, in this setting, dual prices support the optimal allocation, and with high probability, no item is sold out. 

In this paper we are interested in the setting where the number of items is large relative to the supply of each item, that is, $m\gg k$. In this setting, we must account for the possibility that some items will get sold out. When that happens, buyers may shift preferences, increasing the demand for other items. Anticipating and handling this extra demand is a key challenge in multi-item settings. 

Indeed, the gap we exhibit between the single-item and combinatorial settings for static pricings exploits this challenge. We construct a family of instances with just two items and many unit demand buyers, where some buyers prefer item 1 to item 2 and the others prefer 2 to 1. Depending on the prices chosen by the algorithm and the arrival order of the buyers, either item can run out first, causing excess demand for the second item. We show that this results in a lower competitive ratio. Our construction and analysis take inspiration from the lower bound techniques of \citet{jiang2022tightness}. But our setting is greatly complicated by the fact that seller's problem is two-dimensional (as opposed to one-dimensional for \citeauthor{jiang2022tightness}).

On the positive side, \citet{feldman2014combinatorial} handle the challenge of shifting demands in the combinatorial setting with $k=1$ very elegantly by ``splitting'' the welfare of the mechanism equally into the revenue of the seller and the utilities of the buyers. By setting item prices appropriately, they argue that no matter how buyers' preferences and demands shift, the mechanism can either guarantee good revenue for the seller or good utility for the buyers. Unfortunately, due to its very structure, this approach (and its extensions to {\em balanced pricings} in other contexts) cannot obtain a competitive ratio better than $1/2$. 

In order to beat the $1/2$ barrier for $k>1$, we show that by using different prices for different units of the same item, we can refine and extend \citeauthor{feldman2014combinatorial}'s approach to attain balance {\em per copy}. In particular, the price of an item in our supply based pricing increases as more and more units get sold, but the combined contribution of each copy to the revenue and utility is equal in expectation. The seller can therefore guarantee a certain welfare lower bound no matter how many copies of the item get sold. This allows us to achieve a competitive ratio that goes to $1-1/e$ as $k\rightarrow\infty$. Surprisingly, this improvement uses the same information about the underlying instance as \citeauthor{feldman2014combinatorial}'s pricing -- namely, the welfare contribution of every item to the optimal objective.

We then ask whether we can push the competitive ratio even further by allowing for dynamic prices. Because XOS valuations are supported by additive value functions, we show that it is possible to construct an online allocation algorithm that emulates a per-item contention resolution scheme, achieving a competitive ratio of $1-1/\sqrt{k+3}$ -- the same competitive ratio achieved by \citet{alaei2014bayesian} for the single-item setting. 

However, achieving these allocations through item prices is challenging. Restricting the mechanism to item pricings gives the buyers additional control of their own allocation, taking away some of the seller's power. While we are unable to match the performance of general online mechanisms via dynamic item pricing, we develop a novel approach for obtaining a competitive ratio of $1-O(\sqrt{\log k/k})$. We consider the ex-ante LP relaxation with scaled down supply, as in \cite{chawla2017stability}. Our dynamic pricing computes the optimal allocation for this LP; and then for every arriving buyer computes dual prices supporting this intended allocation {\em as a function of the remaining item supply}. We show that due to the structure of XOS valuations, each item individually generates enough social welfare as long as the item has low probability of being sold out.\footnote{ Importantly, unlike \citeauthor{chawla2017stability}'s setting, we only require this probability to be small for any one item at a time, rather than needing a union bound over all items.} We note that if the buyers' value distributions have point masses, this result requires careful tie breaking in the allocation.


\subsection{Further related work.}
The study of prophet inequalities in relation to mechanism design was initiated by \citet{hajiaghayi2007automated}; They presented a $k$-unit static pricing for the single item setting that is asymptotically optimal. Subsequently, the connections between prophet inequalities and mechanism design were further developed by \citet{feldman2014combinatorial} for the social welfare objective and \citet{chawla2010multi} for the revenue objective. In recent years, prophet inequalities in single parameter settings have been studied extensively along multiple directions: e.g., the arrival order of the buyers \cite{Esfandiari2017prophet, Correa2021posted}; sampling-based results \cite{Correa2024sample, Rubinstein2020optimal}; different feasibility constraints \cite{kleinberg2012matroid, dutting2015polymatroid}; non-linear objectives \cite{rubinstein2017combinatorial}, etc. We already discussed the results on multi-unit single-item settings earlier. 

Prophet inequalities for the combinatorial setting have likewise seen much work, but this has largely focused on single-unit supply per item. For XOS buyers, \citet{ehsani2018prophet} show that when buyers arrive in random order, the competitive ratio improves to $1-1/e$, however, this improvement requires setting prices dynamically. There are superficial similarities between our supply based pricing scheme and \citeauthor{ehsani2018prophet}'s time-dependent dynamic pricing scheme in that both adjust prices along a fixed curve as functions of the remaining supply (in our case) and the rate of consumption (in theirs). However, the specific settings call for different styles of analysis. \citet{dutting2024online} further show that a $1/2$-competitive prophet inequality can be constructed for XOS values using polynomially many samples from the value distribution. Surprisingly, this result can be achieved via static item prices. Beyond XOS values, \citet{dutting2020prophet} provided a general framework for designing prophet inequalities based on balanced prices and smoothness. A series of works \cite{feldman2014combinatorial, Zhang2022improved, Dutting2024subadditive} culminating in \citet{correa2023subadditive} showed the existence of a constant factor prophet inequality for subadditive buyers. \citet{Banihashem2024posted} argued that this prophet inequality can be realized through a truthful online mechanism. However, an item pricing based constant-competitive inequality is not yet known for subadditive buyers.

Finally, there is little known about the multi-unit combinatorial setting. To our knowledge, the only prior works to consider this setting are \citet{chawla2017stability} and \citet{chawla2019pricing}, but both make strong restrictive assumptions on the buyers' value functions.
Notably, \citet{chawla2017stability} obtain a competitive ratio of $1-\sqrt{\log k/k}$ via static item prices when items are totally ordered and buyers are unit-demand over intervals.

\section{Preliminaries}

\subsection{Combinatorial auctions and the hindsight optimum}
We consider the standard combinatorial prophet inequality setting with $m$ items and $n$ buyers. Buyers have combinatorial valuations over the items, $v_i:2^{[m]}\rightarrow \Rplus$ for $i\in [n]$, drawn from known independent distributions $\D_i$. We use $\D \coloneqq \D_1 \times \dots \times \D_n$ to denote the joint distribution of values. In the multi-unit setting, each item has several copies available. We use $k_j\in\Z^+$ to denote the number of copies of item $j\in [m]$, $K:= (k_1, \cdots, k_m)$ to denote the supply vector, and $k:=\min_{j\in [m]} k_j$ to denote the minimum multiplicity. The instance is therefore specified by the pair $(\D, K)$. We will usually index buyers by $i$, items by $j$, and a specific copy of an item by $c$.

Our goal is to design an allocation mechanism that maximizes social welfare. For a fixed instantiation of valuation functions $\vals = (v_1, \cdots, v_n)$ where $v_i\sim\D_i$, the (hindsight) optimal social welfare is given by the following integer program:
\begin{align*}
    \hopt(\vals, K):= \max \sum_i v_i(\alloci) & \text{ subject to } \label{Prog:HindOPT} \tag{HOpt}\\
    \sum_{i\in [n]} x_{ij} & \le k_j & \forall j\in [m]\\
    x_{ij} &\in \{0,1\} & \forall i\in [n], j\in [m]
\end{align*}
Here $\alloci$ is the incidence vector of the allocation received by buyer $i$.
Observe that this hindsight optimum can be achieved by the VCG mechanism even without any prior information about the distributions $\D$. We are interested, however, in the simpler class of sequential allocation mechanisms where the mechanism interacts with each buyer in sequence without knowing the instantiated values of future buyers. Our goal is to compete against the expectation over the instantiated values of the hindsight optimum, which is also called the {\em prophet's reward} in a prophet inequality:
\[\prophet(\D,K) := \E_{\vals\sim\D}[\hopt(\vals,K)]\]
We will also consider the stronger ex-ante relaxation benchmark where supply constraints are applied in expectation over the instantiated values. Here $\alloci[i,v_i,S]$ denotes the probability with which buyer $i$ receives a subset $S$ of items when his value is instantiated as $v_i$.
\begin{align*}
    \eaopt(\D, K):= \max \sum_i \sum_{S\subseteq[m]}\E_{v_i\sim\D_i}\left[v_i(S)\,\alloci[i,v_i,S]\right] & \text{ subject to } \label{Prog:EAOPT} \tag{EA-Opt}\\
    \sum_{i\in [n]} \sum_{S\subseteq[m]: S\ni j} \E_{v_i\sim\D_i}\left[\alloci[i,v_i,S]\right] & \le k_j & \forall j\in [m]\\
    \sum_{S\subseteq [m]}\alloci[i,v_i,S] &\le 1 & \forall i\in[n], v_i\in \operatorname{support}(\D_i)\\
    \alloci[i,v_i,S] &\in [0,1] & \forall i\in [n], S\subseteq [m]
\end{align*}
All of our positive (competitive ratio) results are with respect to the stronger ex-ante benchmark, whereas all of our negative (gap) results are with respect to the weaker prophet benchmark.

\subsection{Buyer valuations}

We primarily focus on XOS valuations, as defined below. Our gap results apply to the special case of unit-demand valuations. Observe that unit demand $\subset$ XOS. 

\begin{itemize}
    \item Unit demand: A valuation function $v$ is unit demand if $v(S) = \max_{j\in S} v(\{j\})$ for all subsets $S\subseteq [m]$. Equivalently, the buyer only values receiving one item. We use $v_j$ to denote the value of the $j$th item.
    \item Fractionally subadditive (XOS): A valuation function $v$ is fractionally subadditive if there exists a set $\mathcal{A}^v$ of $m$-dimensional vectors $a \in \R_{\ge 0}^m$, such that $v(S) = \max_{a \in \mathcal{A}^v} \sum_{j \in S} a_j$ for all $S \subseteq [m]$.
\end{itemize}

To ease the presentation and analysis of our schemes, we assume that the distributions are atomless. Our results for dynamic item pricing can be extended to the general case through suitable tie breaking. Our results for supply-based pricing extend immediately to the general case with arbitrary (adversarial) tie breaking.

\subsection{Sequential allocation and pricing}

In this work, we are interested in the social welfare obtained by sequential allocation mechanisms. A sequential allocation mechanism proceeds as follows.

\begin{enumerate}
    \item The instance $(\D, K)$ is revealed and the mechanism $M$ is announced.
    \item Nature draws a valuation profile $v_i \sim \D_i$ for all buyers $i \in [n]$.
    \item An adversary determines the order in which buyers arrive in the mechanism based on $(M, \vals)$. We use $(i)$ to denote the buyer that arrives in the $i$th position. 
    \item At iteration $i \in [n]$, let $R_i$ denote the multiset of items that remains after buyers $(1), \cdots, (i-1)$ have been served. Buyer $(i)$ with valuation $v_{(i)}$ arrives and is allocated a set $S_i:=M(i, R_i, (v_{(1)}, \cdots, v_{(i)}))$ of items, and we set $R_{i+1}:=R_i\setminus S_i$. Note that the allocation $S_i$ can depend on all of the information available to the mechanism $M$ at this iteration. 
    \item At the end of the process, the total social welfare is $\sum_{i \in [n]} v_i(S_i)$.
\end{enumerate}
We denote the expected social welfare of the mechanism $M$ by $\welfare(M, \D, K)$. Observe that in this setting, the choice of the instance $(\D, K)$ as well as the order of arrival of the buyers is chosen adversarially; furthermore, the order can be chosen {\em after} buyers' values have been instantiated. 

We distinguish between several kinds of sequential allocation mechanisms, as follows. 
\begin{itemize}
    \item {\bf Online.} This is the class of all online allocation mechanisms, as defined above. We use $\onl(\D, K):=\max_{M\in\onl}\welfare(M, \D, K)$ to denote the optimal social welfare achieved by this class of mechanisms. 
    \item {\bf Dynamic item pricing.} In a dynamic pricing mechanism, the seller offers each buyer $(i)$ an item pricing $p_i=(p_{ij})_{j\in [m]}$ over the set of remaining items $R_i$. The pricing can depend on all of the information available to the mechanism except the instantiated value of buyer $(i)$, namely, the values $(v_{(1)}, \cdots, v_{(i-1)})$ and the set $R_i$. Buyer $(i)$ purchases the set of items that maximizes her utility: $S_i:=\argmax_{S\subseteq R_i} \{v_i(S)-\sum_{j\in S} p_j\}$. We use $\dyn(\D, K):=\max_{\text{dynamic pricing } p}\welfare(p, \D, K)$ to denote the optimal social welfare achieved by this class of mechanisms.
    \item {\bf Static item pricing.} In a static pricing mechanism, the seller determines a fixed price $p_j$ for each item $j\in [m]$ upfront. The prices can depend on the instance $(\D, K)$ but do not depend on the instantiation of values or the order of arrival of the buyers. When buyer $(i)$ arrives, he is offered the pricing $p=(p_j)_{j\in [m]}$ over the set of remaining items $R_i$ and purchases $S_i:=\argmax_{S\subseteq R_i} \{v_i(S)-\sum_{j\in S} p_j\}$. We use $\stat(\D, K):=\max_{\text{static pricing } p}\welfare(p, \D, K)$ to denote the optimal social welfare achieved by this class of mechanisms.
    \item {\bf Supply-based static pricing.} In this mechanism, the seller determines upfront a fixed price vector that assigns a (potentially different) price $p_{j,c}$ to every copy $c$ of every item $j$. When buyer $(i)$ arrives, if $c_j$ copies of item $j$ have been sold so far, the buyer is offered item pricing $(p_{j, c_j+1})_{j\in [m]}$, where $p_{j, k_j+1}$ is understood to be $\infty$ (corresponding to the item having sold out). We use $\copyip(\D, K):= \max_{\text{supply-based pricing } p}\welfare(p, \D, K)$ to denote the optimal social welfare achieved by this class.
\end{itemize}
Observe that static and supply based pricing are anonymous and non-adaptive. Dynamic pricing, on the other hand, can be both non-anonymous and adaptive. The following relationships are immediate:
\[ \eaopt(\D, K) \ge \prophet(\D, K) \ge  \onl(\D, K) \ge \dyn(\D, K) \ge \copyip(\D, K)\ge \stat(\D, K). \]
We remark that any online allocation mechanism can be implemented as a (truthful) sequential pricing mechanism without hurting its social welfare, where buyer $(i)$ is offered a pricing over sets of items and can choose his favorite set to buy under this pricing \cite{Banihashem2024posted}.


\subsection{Performance metric}

The \emph{competitive ratio} of a mechanism $M$ is the worst-case welfare-to-optimum ratio across all the possible set of distributions $D$ and supply vectors $K$, and is expressed as a function of the minimum multiplicity $k:=\min_{j\in [m]} k_j$: 
\[\compratio(M, k) \coloneqq \inf_{\D, K: k_j\ge k\forall j\in [m]} \frac{\welfare(M, \D, K)}{\eaopt(\D, K)}\]
We can further define the competitive ratio of a class of mechanisms as the ratio above, where the numerator $\welfare(M, \D, K)$ is replaced by the maximum welfare achieved by the respective class over the given instance.

\section{The competitive ratio of supply-based static pricing}\label{sec:pricing}

In this section we show that for XOS buyers supply based pricing can achieve a competitive ratio that is a strictly increasing function of the supply $k$. The main theorem of this section is as follows.

\begin{theorem}
\label{thm:approx}
    Given any instance $(\D, K)$ and a feasible solution $\allocs=\{\alloci[i,v_i,S],\forall i\in [n],v_i\}$ to the linear program~\eqref{Prog:EAOPT}, we can efficiently compute 
    a supply-based static pricing mechanism such that, the expected welfare of the mechanism is at least \[\left(1-\left(\frac{k}{k+1}\right)^k\right)\sum_i \sum_{S\subseteq[m]}\E_{v_i\sim\D_i}\left[v_i(S)\,\alloci[i,v_i,S]\right]\]
\end{theorem}

Before proving Theorem~\ref{thm:approx}, let's first discuss its implications. The result provides for {\em any} feasible solution $\allocs=\{\alloci[i,v_i,S],\forall i\in [n],v_i\}$ of the linear program~\eqref{Prog:EAOPT}, an efficiently computable supply-based pricing mechanism that gives a $1-(\frac{k}{k+1})^k$ approximation to the social welfare obtained by that feasible LP solution. By plugging in the optimal solution $\allocs^*=\{\alloci^*[i,v_i,S],\forall i\in [n],v_i\}$ of the program~\eqref{Prog:EAOPT}, we obtain a $1-\left(\frac{k}{k+1}\right)^k$ approximation to $\eaopt(\D, K)$. We therefore obtain the following corollary:
\begin{corollary}\label{cor:opt}
    For any $k\ge 1$ the competitive ratio of supply based static pricing for the multi-unit combinatorial auction setting with XOS buyers is \[\compratio(\copyip, k) = 1-\left(\frac{k}{k+1}\right)^k\]
\end{corollary}

As we discussed previously, the supply based pricing that achieves the above results is computed using the contribution of every item $j$ to the objective of \eqref{Prog:EAOPT}. We show that the competitive ratios stated above are tight in a certain sense: no supply-based pricing mechanism that uses this limited information can obtain a better performance. We exhibit this tightness even in single-item multi-unit settings.
\begin{theorem}
\label{thm:supply-tightness}
    There exists a family $\F$ of single-item $k$-unit instances such that:
    \begin{itemize}
        \item For every instance $(\D,k)\in \F$, the value of \eqref{Prog:EAOPT} is $2k$.
        \item For every supply-based pricing vector $p=(p_1, \cdots, p_k)$, there exists an instance $(\D_p, k)\in \F$ with $\welfare(p, \D_p, k) \le (1-(k/k+1)^k)\cdot 2k$.
    \end{itemize}
    Consequently, a supply based pricing mechanism that knows the value of \eqref{Prog:EAOPT} for a given instance but not the value distribution cannot obtain a competitive ratio better than $1-(k/k+1)^k$.
\end{theorem}

We devote the rest of this section to proving Theorem~\ref{thm:approx}. A proof of Theorem~\ref{thm:supply-tightness} can be found in Appendix~\ref{sec:supp-tight}.

\vspace{0.1in}
\textsc{Proof of Theorem~\ref{thm:approx}.} Before we go into the details of the proof, we first introduce an outline and some useful notation. For a given feasible solution $\allocs$, we first compute the contribution $\sw_j$ of any item $j$ to the total social welfare under $\allocs$ by taking out the additive representative function for any set $S$ being allocated. We want to recover some fraction of this contribution from item $j$ in our supply based pricing. 

As in the approach of \citet{feldman2014combinatorial}, we partition the social welfare achieved by the pricing into its revenue and utility components: every time item $j$ is sold, its revenue can be attributed to $j$'s contribution to the social welfare. On the other hand, every time that item $j$ is available but not claimed by a buyer that the solution $\allocs$ sells it to, we can argue that the buyer obtains good utility by purchasing an alternative set of items. We again attribute this utility to $j$'s contribution to the social welfare. 

Our key observation that sets it apart from the analysis of \citeauthor{feldman2014combinatorial} is to track and control how the revenue and utility change as functions of the number of copies of an item sold. When all copies of the item have the same prices, the revenue of an item increases linearly with each extra copy sold. On the other hand, the total utility of the buyers from this item remains unchanged as long as at least one copy of the item remains unsold, and immediately drops to $0$ when the item gets sold out. As a result the total contribution of the item at first increases and then suddenly drops as more and more copies get sold. We set prices on the copies in such a manner that the total contribution of each copy becomes equal. In doing so, the total social welfare obtained by the pricing becomes independent of the arrival order and trajectory of the algorithm, and we can obtain tighter bounds. Specifically, we set increasing prices on each successive copy in such a manner that revenue increases as a convex function of the number of copies sold, while utility decreases gradually as a concave function.



We now describe the details.
Fix a buyer $i$, a valuation $v_i$, and a subset of items $S\subseteq [m]$. Recall that we define $a^{v_i,S}\in\mathcal{A}^{v_i}$ to be the supporting additive function for the set $S$. In other words, we have: (i) $v_i(S')\ge \sum_{j\in S'}a_j^{v_i,S}$ for any $S'\subseteq [m]$, and (ii) $v_i(S)= \sum_{j\in S}a_j^{v_i,S}$. For item $j\in S$, $a_j^{v_i,S}$ denotes the contribution of item $j$ to the buyer's value of set $S$. We define $\sw_j$, representing the contribution of item $j$ to the total social welfare (under the feasible solution $\allocs$) as follows:
\[\sw_j=\sum_i \sum_{S\subseteq[m]: S\ni j}\E_{v_i\sim\D_i}\left[ a_j^{v_i,S}\,\alloci[i,v_i,S]\right]\]

Notice that we can write the total social welfare under the feasible solution $\allocs$, denoted as $\sw$, as the sum of the contribution of all the items:
\[\sw=\sum_{j\in[m]}\sw_j\]

Now we consider our supply-based static pricing scheme. For $c\in [k_j]$ we set the price of the $c$-th copy of the item $j$ to be $\alpha_{j,c} \cdot \sw_j$, where \[\alpha_{j,c}=\frac{1}{k_j}\left(\frac{k_j}{k_j+1}\right)^{k_j+1-c}.\] 
Note that $\sw_j/k_j$ is the {\em per-unit} contribution of $j$ to the social welfare of $\allocs$. Our supply-based pricing charges prices that are some fixed fraction of this per-unit price, with the fraction getting exponentially closer to $1$ with each successive copy. For example, for $k_j=2$ these fractions are $4/9$ and $2/3$, and for $k_j=3$, these fractions are $27/64$, $9/16$ and $3/4$.

Consider any arbitrary ordering over buyer arrivals, and let $q_{j, c}$ denote the probability over the instantiations of buyers' valuations that the above price scheme sells {\em exactly} $c$ copies of item $j$.  Now let us analyze the performance of our pricing scheme. We will use $\rev_j$ to denote the revenue we obtained from selling item $j$. We have
\begin{align}
    \rev_j = \sum_{c = 1}^{k_j} \left(\alpha_{j,c} \cdot \sw_j \cdot \sum_{l = c}^{k_j} q_{j, l}\right)\label{eq:revj}
\end{align}
where we use the fact that we are selling the $c$-th copy at a price of $\alpha_{j,c} \cdot \sw_j$ with probability $\sum_{l = c}^{k_j} q_{j, l}$.


We will now estimate the contribution of the buyers' utilities in expectation over $\vals\sim\D$. 
Fix a buyer $i$ and valuation $v_i\sim\D_i$. Recall that the fractional solution $\allocs$ provides a probability distribution over allocations to this buyer: $\{x_{i, v_i, S}\}$ for $S\subseteq [m]$. Let us draw a set $S$ from this distribution. Now consider a run of our supply based pricing scheme. Suppose that when buyer $i$ arrives, $C_j$ copies of item $j$ have been sold, where $C_j$ is a random variable. Then, the buyer can obtain the following utility from purchasing set $S$:
\[ \sum_{j\in S: C_j< k_j} \max\{a_j^{v_i,S}-\alpha_{j, C_j+1} \cdot \sw_j, 0\}\]
Taking expectations over the draws of $v_i$ and $S$, and recalling that $C_j$ is independent of these choices, we get that the utility of buyer $i$ is at least:
\begin{align*}
& \E_{v_i\sim\D_i}\left[ \sum_{S\subseteq [m]} x_{i, v_i, S} \sum_{j\in S: C_j< k_j} \max\{a_j^{v_i,S}-\alpha_{j, C_j+1} \cdot \sw_j, 0\}\right] \\
& = \sum_{S\subseteq [m]} \sum_{j\in S: C_j< k_j} \E_{v_i\sim\D_i}\left[ x_{i, v_i, S} \max\{ a_j^{v_i,S} - \alpha_{j, C_j+1} \cdot \sw_j, 0\} \right]
\end{align*}
Observe that if we replace the random variable $C_j$ with another variable that first order stochastically dominates it, then the summand inside the expectation decreases, as the only term in the summand that depends on $C_j$ is the price $\alpha_{j,C_j+1} \cdot \sw_j$. On the other hand, the number of terms in the sum over $j$ also decreases. Therefore, increasing $C_j$ decreases the entire expression above. Let $\hat{C}_j$ denote the number of copies of item $j$ sold by the supply based pricing {\em at the end of the process}, once all buyers have arrived. Then, $C_j\le \hat{C}_j$ with probability $1$. We can therefore replace $C_j$ by $\hat{C}_j$ and obtain a lower bound on the buyer's utility. Taking expectation over $\hat{C}_j$ we get:
\begin{align*}
    u_i & \ge \E_{\hat{C}_1, \cdots, \hat{C}_n}\left[\sum_{S\subseteq [m]} \sum_{j\in S: \hat{C}_j< k_j} \E_{v_i\sim\D_i}\left[ x_{i, v_i, S} \max\{ a_j^{v_i,S} - \alpha_{j, \hat{C}_j+1} \cdot \sw_j, 0\} \right]\right]\\
    & = \sum_{S\subseteq [m]} \sum_{j\in S} \sum_{c=0}^{k_j-1} q_{j,c}\, \E_{v_i\sim\D_i}\left[ x_{i, v_i, S} \max\{ a_j^{v_i,S} - \alpha_{j, c+1} \cdot \sw_j, 0\} \right]\\
    & \ge \sum_{j\in [m]} \sum_{S\subseteq [m]: S\ni j} \sum_{c=0}^{k_j-1} q_{j,c}\,\left(\E_{v_i\sim\D_i}\left[ x_{i, v_i, S} a_j^{v_i,S}\right] - \E_{v_i\sim\D_i}\left[ x_{i, v_i, S} \right]\alpha_{j, c+1} \cdot \sw_j\right)
\end{align*}
where the second line replaces the expectation over the $\hat{C}_j$'s with a sum over their possible values and respective probabilities and the last line follows by removing the max and rearranging the sum. 

We now sum up the utilities of all of the agents $i\in [n]$ and write the utility per item $j\in [m]$ as:
\begin{align*}
    \uti_j & \ge \sum_{i\in [n]} \sum_{S\subseteq [m]: S\ni j} \sum_{c=0}^{k_j-1} q_{j,c}\,\E_{v_i\sim\D_i}\left[ x_{i, v_i, S} a_j^{v_i,S}\right] - \sum_{i\in [n]} \sum_{S\subseteq [m]: S\ni j} \sum_{c=0}^{k_j-1} q_{j,c}\,\E_{v_i\sim\D_i}\left[ x_{i, v_i, S} \right]\alpha_{j, c+1} \cdot \sw_j\\
    & = \sum_{c=0}^{k_j-1} q_{j,c}\left( \sw_j - \left(\sum_{i\in [n]} \sum_{S\subseteq [m]: S\ni j} \E_{v_i\sim\D_i}\left[ x_{i, v_i, S}\right] \right) \alpha_{j, c+1} \cdot \sw_j\right)
\end{align*}
Here we used the definition of $\sw_j$ to simplify the first sum and then combined terms together.
Finally, we recall that by the first constraint in \eqref{Prog:EAOPT}, the inner double sum is at most $k_j$. Therefore we get:
\begin{align}
    \uti_j \ge \sum_{c=0}^{k_j-1} q_{j,c}\sw_j \left( 1-k_j \alpha_{j, c+1}\right) \label{eq:utij}
\end{align}
Finally, we add the revenue and utility contributions of item $j$, namely \eqref{eq:revj} and \eqref{eq:utij}, to obtain:
\begin{align*}
    \alg_{j}&=\rev_j+\uti_j\\
    &\ge \sw_j \cdot \sum_{c=1}^{k_j}\left( \alpha_{j,c} \sum_{l=c}^{k_j} q_{j,l} \right)+ \sw_j \cdot \sum_{c=0}^{k_j-1} q_{j,c}(1-k_j\alpha_{j,c+1}) 
\end{align*}
We can aggregate the $q_{j,c}$ terms as follows where we have set $\alpha_{j, k_j+1}=1/k_j$.
\begin{align*}
    \alg_j \ge \sw_j\sum_{c=0}^{k_j} q_{j,c} \left( 1-k_j\alpha_{j,c+1} + \sum_{l=1}^{c} \alpha_{j,l}  \right) 
\end{align*}
Observe that by our choice of setting $\alpha_{j,c}:= \frac{1}{k_j}(k_j/(k_j+1))^{k_j+1-c}$, each of the terms multiplied by the probabilities $q_{j,c}$ above is equal to $1- (k_j/(k_j+1))^{k_j}$. We therefore get:
\[\alg_j\ge \left(1-\frac{k_j}{k_j+1}\right)^{k_j}\sw_j\]

Since our pricing scheme obtain a social welfare of $\sum_j \alg_j$ while the  welfare obtained by $\allocs$ is $\sum_j \sw_j$, we obtain a competitive ratio of $1-(\frac{k}{k+1})^{k}$ where $k=\min_{j\in [m]} k_j$. This completes the proof of the theorem. \hfill $\square$

\section{The competitive ratio of dynamic item pricing}
\label{sec:dynamic}

In this section we show that for the multi-unit combinatorial setting
$\compratio(\dyn, k) = 1 - O(\sqrt{\log{k}/k})$,  asymptotically matching the competitive ratio of  static pricing for the single-item multi-unit setting.



\begin{theorem}
    \label{thm:dyn-xos}
    Given any instance $(\D, K)$ and a feasible solution $\allocs=\{\alloci[i,v_i,S],\forall i\in [n],v_i\}$ to the linear program~\ref{Prog:EAOPT}, there exists a dynamic pricing mechanism such that, the expected welfare of the mechanism is at least $(1 - O(\sqrt{\log k/k})) \cdot \sum_i \sum_{S\subseteq[m]}\E_{v_i\sim\D_i}\left[v_i(S)\,\alloci[i,v_i,S]\right]$, where $k = \min_j k_j$.
\end{theorem}

\begin{algorithm}[tbh]
\caption{Dynamic Pricing for Combinatorial Prophet Inequality}
\label{alg:xos-dyn}
\KwIn{Combinatorial prophet inequality instance $(\D, K)$, feasible solution $\{\alloci[i, v_i, S]\}$} 
For every $i, v_i, S$, set $\talloci[i, v_i, S] = (1 - C \cdot \sqrt{\log k/k}) \cdot \alloci[i, v_i, S]$ with a sufficiently large constant $C$, where $k = \min_j k_j$. \\
For every $i, j$, define $z_{i, j} = \E[\sum_{S \ni j} \talloci[i, v_i, S]]$ to be the expected amount of item $j$ consumed by buyer $i$. \\
\For{$i = 1 \to n$}
{
    Let $R_i \subseteq [m]$ be the subset of items with at least one copy left. \\ 
    Solve the following LP and let $\{\ysalloci[v_i, S]\}$ be the optimal solution:
    \begin{align*}
        \max  \sum_{S\subseteq[R_i]}\sum_{v_i \in \operatorname{support}(\D_i)} v_i(S)& \cdot \yalloci[v_i, S]  \text{ subject to } \label{Prog:ssopt} \\
         \sum_{S\subseteq[R_i]: S\ni j} \sum_{v_i \in \operatorname{support}(\D_i)}\yalloci[v_i,S] & \le z_{i, j} & \forall j\in R_i\\
        \sum_{S\subseteq [R_i]}\yalloci[v_i,S] &\le \pr\left[v_i \text{ is realized}\right] & \forall v_i\in \operatorname{support}(\D_i)\\
        \yalloci[v_i,S] &\in [0,1] & \forall S\subseteq [R_i]
    \end{align*}\\
    Properly set prices $p^{(i)}_1, \cdots, p^{(i)}_m$ and a tie-breaking rule, so that 
    \[
    \pr\left[v_i \text{ is realized} ~\land ~ S = \arg \max_{S} \left(v_i(S) - \sum_{j \in S} p^{(i)}_j \right)\right] ~=~ \ysalloci[v_i, S].
    \]\\
    Realize $v_i \sim \D_i$ and allocate $S_i \in \arg \max_S (v_i(S) - \sum_{j \in S} p^{(i)}_j)$ to buyer $i$. Break the tie according to the tie-breaking rule stated in Line 6.
}
\end{algorithm}

We provide a constructive proof of \Cref{thm:dyn-xos}  through \Cref{alg:xos-dyn} described below. The main idea of \Cref{alg:xos-dyn} is as follows: we ask the algorithm to follow the marginal probability of each item provided by $\{\alloci[i, v_i, S]\}$ after scaling down by a $1 - O(\sqrt{\log k/k})$ factor. At each time step, we solve the LP in Line 5 and get the optimal solution $\{\ysalloci[v_i, S]\}$. Then, the following \Cref{lma:dyn-prices} first suggests that there exists a price vector that captures the optimal solution $\{\ysalloci[v_i, S]\}$, i.e., the probability that $v_i$ is realized and $S$ is the favorite bundle is $\ysalloci[v_i, S]$.

\begin{lemma}
\label{lma:dyn-prices}
    Given $\{\ysalloci[y_i, S]\}$ being the optimal solution of the linear program in Line 5 of \Cref{alg:xos-dyn}, there exist prices $\{p^{(i)}_1\}_j$ and a tie-breaking rule, which are independent to the realization of $v_i$, that satisfy 
    \[
    \pr\left[v_i \text{ is realized} ~\land ~ S = \arg \max_{S} \left(v_i(S) - \sum_{j \in S} p^{(i)}_j \right)\right] ~=~ \ysalloci[v_i, S].
    \]
\end{lemma}

\Cref{lma:dyn-prices} is based on duality. We defer the proof to \Cref{sec:dyn-proofs}. The lemma says that in expectation the welfare achieved by the prices matches the welfare given by $\{\ysalloci[v_i, S]\}$. It remains to show that $\{\ysalloci[v_i, S]\}$ achieves a good welfare. To begin, we first give the following lemma, which guarantees that any individual item has low probability of getting sold out. Recall that $R_{n+1}$ is the set of items left over after all buyers have arrived.


\begin{lemma}
    \label{lma:item-remain}
    For each $j \in [m]$, we have $\pr[j \in R_{n+1}] \geq 1 - k^{-2}$. 
\end{lemma}

To prove \Cref{lma:item-remain}, we need concentration on a martingale process. The following inequality provides what we need for the concentration:

\begin{theorem}[Theorem A of \cite{fan2015exponential}]
\label{thm:concentration}
Assume that we are given a sequence of real-valued supermartingale differences $(\xi, \mathcal{F}_i)_{i = 0, \cdots, n}$ defined on some probability space with $\xi_0 = 0$ and  $\mathcal{F}_0 \subseteq \mathcal{F}_1 \subseteq \cdots \subseteq\mathcal{F}_n$ are increasing $\sigma$-fields. Provided that $\E[\xi_i|\mathcal{F}_{i-1}] \leq 0$ and $|\xi_i| \leq 1$ for $i \in [n]$, and $\sum_{i \in [n]} \E[\xi_i^2|\mathcal{F}_{i-1}] \leq v^2$, we have
\[
\pr\left[\sum_{i \in [n]} \xi_i^2 \geq x\right] \leq \exp\left(- \frac{x^2}{v^2 + x}\right).
\]
\end{theorem}

Now, we prove \Cref{lma:item-remain} via \Cref{thm:concentration}.
\begin{proof}[Proof of \Cref{lma:item-remain}]
    Fix $j$. It's sufficient to show that with probability at least $1 - k^{-2}$, \Cref{alg:xos-dyn} allocates at most $k_j$ copies of item $j$. 

    We apply prove the above statement via \Cref{thm:concentration}. Let random variable $\eta_i \in \{0, 1\}$ be the number of item $j$ we allocate to buyer $i$, and $\xi_i = \eta_i - z_{i, j}$. Both $\eta_i$ and $\xi_i$ depend on the realization of $R_i$. Note that for any realization of $R_i$, we have
    \[
    \E[\xi_i |R_i] ~=~ \E[\eta_i | R_i] - z_{i, j} ~\leq~ 0,
    \]
    where the inequality follows from the fact that \Cref{lma:dyn-prices} guarantees that the dynamic pricing algorithm allocates subset $S$ to buyer $i$ with probability $\sum_{v_i \in \operatorname{support}(\D_i)} \ysalloci[v_i, S]$, while the linear program in \Cref{alg:xos-dyn} guarantees $\sum_{S \ni j} \sum_{v_i \in \operatorname{support}(\D_i)} \ysalloci[v_i, S] \leq z_{i, j}$. For the second moment constraint, for any realization of $R_i$, we have
    \begin{align*}
        \E[\xi^2_i |R_i] ~&=~ \E[\eta^2_i |R_i] + z^2_{i, j} - 2 z_{i, j} \cdot \E[\eta |R_i] \\
        ~&=~\E[\eta_i |R_i] + z^2_{i, j} - 2 z_{i, j} \cdot \E[\eta |R_i] ~\leq~ z_{i, j},
    \end{align*}
    where the second equality uses the fact that $\eta_i \in \{0, 1\}$, and the inequality follows from simple algebra together with the fact that $0 \leq \E[\eta_i |R_i] \leq z_{i, j} \leq 1$. Summing the above inequality over $i \in [n]$, we have
    \[
    \sum_{i \in [n]} \E[\xi^2_i |R_i] ~\leq~ \sum_{i \in [n]} z_{i, j}  ~\leq~ k_j.
    \]
    Now, we apply \Cref{thm:concentration}. Recall that we aim at showing $\pr\left[\sum_{i \in [n]} \eta_j > k_j\right] \leq 1/k^2$, which follows from
    \begin{align*}
        \pr\left[\sum_{i \in [n]} \eta_i > k_j\right] ~&\leq~ \pr\left[\sum_{i \in [n]} \xi_i > \frac{C \log k}{\sqrt{k}} \cdot k_j\right] \\
        ~&\leq~ \exp\left(- \frac{k^2_j \cdot \frac{C^2 \log k}{k}}{k_j + k_j \cdot \frac{C \sqrt{\log k}}{\sqrt{k}}}\right) \\
        ~&\leq~ \exp(-C \log k /2) ~\leq~ \frac{1}{k^2},
    \end{align*}
    where the first inequality uses the fact that $k_j - \sum_{i \in [n]} z_{i, j} \geq k_j \cdot \frac{C \sqrt{\log k}}{\sqrt{k}}$, and the second inequality applies \Cref{thm:concentration} with $x = k_j \cdot \frac{C \sqrt{\log k}}{\sqrt{k}}$ and $v^2 = k_j$,  the third inequality uses $k_j \geq k$ and $\sqrt{\log k/k} \leq 1$, and the last inequality holds when $C$ is sufficiently large.
\end{proof}

Next, we prove \Cref{lma:welfare-comparable}, which shows that the welfare given by $\{\ysalloci[v_i, S]\}$ is comparable to the welfare given by $\{\talloci[i, v_i, S]\}$.

\begin{lemma}
    \label{lma:welfare-comparable}
    For every $i \in [n]$, we have
    \[
    \E_{R_i}\left[\sum_{S \subseteq R_i} \sum_{v_i \in \operatorname{support}(\D_i)} v_i(S) \cdot \ysalloci[v_i, S]\right] ~\geq~ \left(1 - \frac{1}{k^2}\right) \cdot \E_{v_i}\left[\sum_{S \subseteq[m]} \talloci[i, v_i, S] \cdot v_i(S) \right].
    \]
\end{lemma}

\begin{proof}
    Fix $R_i$. Consider the following process:
    \begin{itemize}
        \item Draw $v_i \sim \D_i$, and then   draw $S$ with probability $\talloci[i, v_i, S]$.
        \item Allocate $S \cap R_i$ to buyer $i$.
    \end{itemize}
    Let $\tilde y_{v_i, S}$ be the probability that $v_i$ is realized in the above process, and subset $S$ is allocated to buyer $i$. Then, we have
    \begin{align*}
        &~\sum_{S \subseteq R_i} \sum_{v_i \in \operatorname{support}(\D_i)} v_i(S) \cdot \tilde y_{v_i, S} \\
        ~=&~ \sum_{v_i \in \operatorname{support}(\D_i)} \pr\left[v_i \text{ is realized}\right] \cdot \sum_{S \subseteq R_i} \talloci[i, v_i, S] \cdot v_i(S \cap R_i) \\
        ~\geq&~ \sum_{v_i \in \operatorname{support}(\D_i)} \pr\left[v_i \text{ is realized}\right] \cdot \sum_{S \subseteq R_i} \talloci[i, v_i, S] \cdot \sum_{j \in S \cap R_i} a_j^{v_i, S} \\
        ~=&~ \sum_{j \in [m]} \one[j \in R_i] \cdot \E_{v_i}\left[ \sum_{S \subseteq[m]} \talloci[i, v_i, S] \cdot a^{v_i, S}_j \right],
    \end{align*}
    where the inequality follows from the property of XOS functions. 
    Now take the expectation over the randomness of $R_i$ for the above inequality. Since \Cref{lma:item-remain} guarantees that $j \in R_i$ with probability at least $1 - k^{-2}$, we have
    \begin{align*}
        \E_{R_i}\left[\sum_{S \subseteq R_i} \sum_{v_i \in \operatorname{support}(\D_i)} v_i(S) \cdot \tilde y_{v_i, S}\right] ~&\geq~ \left(1 - \frac{1}{k^2}\right) \cdot \sum_{j \in [m]} \E_{v_i}\left[\sum_{S \subseteq[m]} \talloci[i, v_i, S] \cdot a^{v_i, S}_j \right] \\
        ~&=~ \left(1 - \frac{1}{k^2}\right) \cdot \E_{v_i}\left[\sum_{S \subseteq[m]} \talloci[i, v_i, S] \cdot v_i(S) \right].
    \end{align*}
    Finally, note that the optimality of $\{\ysalloci[v_i, S]\}$ guarantees that the LHS of first line is upper bounded by the welfare given by $\ysalloci[v_i, S]$. Therefore, we have
    \[
    \E_{R_i}\left[\sum_{S \subseteq R_i} \sum_{v_i \in \operatorname{support}(\D_i)} v_i(S) \cdot \ysalloci[v_i, S]\right] ~\geq~ \left(1 - \frac{1}{k^2}\right) \cdot \E_{v_i}\left[\sum_{S \subseteq[m]} \talloci[i, v_i, S] \cdot v_i(S) \right]. \qedhere
    \]
\end{proof}

Now, we are ready to prove \Cref{thm:dyn-xos}.

\begin{proof}[Proof of \Cref{thm:dyn-xos}]
    Summing the inequality in \Cref{lma:welfare-comparable} for all $i \in [n]$, we have
    \[
    \sum_{i \in [n]} \E_{R_i}\left[\sum_{S \subseteq R_i} \sum_{v_i \in \operatorname{support}(\D_i)} v_i(S) \cdot \ysalloci[v_i, S]\right] ~\geq~ \left(1 - \frac{1}{k^2}\right) \cdot \sum_{i \in [n]} \sum_{S \subseteq [m]} \E_{v_i}\left[ \talloci[i, v_i, S] \cdot v_i(S) \right].
    \]
    Since \Cref{lma:dyn-prices} guarantees  the probability that $v_i$ is realized and subset $S$ is allocated to buyer $i$ is $\ysalloci[v_i, S]$, the LHS of above inequality represents the expected welfare gained by \Cref{alg:xos-dyn}. For the RHS of the above inequality, we have
    \begin{align*}
        \left(1 - \frac{1}{k^2}\right) \cdot \sum_{i \in [n]} \E_{v_i}\left[ \talloci[i, v_i, S] \cdot v_i(S) \right] ~&=~ \left(1 - \frac{1}{k^2}\right) \cdot \left(1 - \frac{C \cdot \sqrt{\log k}}{\sqrt{k}} \right) \cdot \sum_{i \in [n]} \sum_{S \subseteq [m]} \E_{v_i}\left[ \alloci[i, v_i, S] \cdot v_i(S) \right] \\
        ~&=~ \left(1 - O(\sqrt{\log k/k}) \right) \cdot \sum_{i \in [n]} \sum_{S \subseteq [m]} \E_{v_i}\left[ \alloci[i, v_i, S] \cdot v_i(S) \right],
    \end{align*}
    which finishes the proof of \Cref{thm:dyn-xos}.
\end{proof}

\section{Unit-Demand is Harder than Single-Item}\label{sec:hard-instance}

In this section, we demonstrate that the combinatorial multi-unit setting is strictly harder than the single-item multi-unit setting in the sense that static pricing obtains a competitive ratio strictly smaller in the former setting than in the latter setting. 

Let $\tau_k$ denote the tight bound on the competitive ratio of static pricing for the single item setting, where $k$ denotes the item supply, as established by \citet{jiang2022tightness}. We establish a bound $\hat{\tau}_k$ on the competitive ratio of static item pricing in the combinatorial setting as a function of the item supply $k$. Both the quantities $\tau_k$ and $\hat{\tau}_k$ are solutions to one-dimensional equations, as described below. While the equations do not have closed-form solutions, the values can be computed numerically using the bisection method such that at least six digits after the decimal point are correct. We verify through numerical evaluation of the bounds that $\hat{\tau}_k < \tau_k$ for $k$ up to at least $1000$, with the difference between the two quantities appearing within the first four decimal digits. We conjecture that $\hat{\tau}_k < \tau_k$ for every $k \ge 2$, but were not able to prove this formally.



We now elaborate on the details. The following theorems establish expressions for $\tau_k$ and $\hat{\tau}_k$ respectively.

\begin{theorem}[\citet{chawla2024static, jiang2022tightness}]\label{thm:worst-case-for-single-item}
    Fix any $k \in \Z^+$. For a rate $\lambda \in [0, \infty)$, we define $\mu_k(\lambda) = \frac{\E[\min\{\Pois(\lambda), k\}]}{k}$ and $\delta_k(\lambda) = \pr[\Pois(\lambda) < k]$. Let $\lambda^*_k$ be the unique root to the equation $\mu_k(\lambda_k^*) = \delta_k(\lambda_k^*)$. For every instance of the $k$-unit single-item prophet inequality, there exists a static pricing scheme that achieves the competitive ratio of
    \[\tau_k \coloneqq \mu_k(\lambda_k^*) = \delta_k(\lambda_k^*).\]
    Furthermore, this ratio is tight, meaning that for every $\epsilon > 0$, there exists an instance where the competitive ratio of any static pricing scheme is at most $\tau_k + \epsilon$.
\end{theorem}


\begin{theorem}\label{thm:worst-case-for-two-items}
    Fix any $k \in \Z^+$. For a rate $\lambda \in [0, \infty)$, we define $\hat{\mu}_k(\lambda) = \frac{\E[\min\{\Pois(\lambda), 2k\}]}{2k}$ and $\hat{\delta}(\lambda) = \sum_{i=0}^{2k - 1} \pr[\Pois(\lambda) = i] \cdot \pr[\Binom(i, 1/2) < k]$. Let $\hat{\lambda}^*_k$ be the unique root to the equation $\hat{\mu}_k(\hat{\lambda}^*_k) = \hat{\delta}_k(\hat{\lambda}^*_k)$. Let $\hat{\lambda}'_k$ be the unique root to the equation $\left(\frac{d}{d\lambda}\hat{\mu}_k(\hat{\lambda}'_k)\right) + \left(\frac{d}{d\lambda}\hat{\delta}_k(\hat{\lambda}'_k)\right) = 0$. Let
    \[\hat{\tau}_k \coloneqq \frac{1}{2} \left(\hat{\mu}_k(\max\{\hat{\lambda}_k^*, \hat{\lambda}'_k\}) + \hat{\delta}_k(\max\{\hat{\lambda}_k^*, \hat{\lambda}'_k\})\right).\]
    
    Then for every $\epsilon > 0$, there exists an instance where the competitive ratio of any static pricing scheme is at most $\hat{\tau}_k + \epsilon$.
\end{theorem}
The following table lists the values of $\tau_k$ and $\hat{\tau}_k$ for $k\le 11$.


\begin{table}[h!]
	\begin{minipage}{\columnwidth}
		\begin{center}
\begin{tabular}{|l|l|l|l|l|l|l|l|l|l|l|}
\hline
$k$            & 2      & 3      & 4      & 5      & 6      & 7      & 8      & 9      & 10     & 11     \\ \hline
$\tau_k$       & 0.5859 & 0.6309 & 0.6605 & 0.6821 & 0.6989 & 0.7125 & 0.7239 & 0.7337 & 0.7422 & 0.7497 \\ \hline
$\hat{\tau}_k$ & 0.5843 & 0.6286 & 0.6578 & 0.6793 & 0.6960 & 0.7096 & 0.7210 & 0.7307 & 0.7392 & 0.7468 \\ \hline
\end{tabular}
		\end{center}
	\end{minipage}
	\caption{$\tau_k$ versus $\hat{\tau}_k$ for $k$ from $2$ to $11$.}
\end{table}%

The rest of this section is devoted to proving Theorem~\ref{thm:worst-case-for-two-items}. We describe the hard instance in detail and present an outline of our analysis, with proofs deferred to Appendix~\ref{app:hard-instance}.

\subsection{Describing the hard instance for single-item}\label{sec:one-item}

We briefly discuss the hard instance for a single item, which will motivate our hard instance for the two-item unit-demand case. Consider $n + 1$ buyers, where $n$ of them has value that is drawn i.i.d. from $\Unif[1, 1 + \epsilon]$, while the last buyer has value $\frac{U_k}{\epsilon}$ with probability $\epsilon$, and $0$ otherwise, and $U_k$ to be defined later. Here, we let $n \to \infty$, while $\epsilon \to 0$.

Let's call the first $n$ buyers \emph{small buyers}, while the last buyer is the \emph{large buyer}. Ignoring $O(\epsilon)$ terms, observe that the prophet gets at least $k + U_k$ social welfare on expectation. On the other hand, for any static price $p \in [1, 1 + \epsilon]$ set for the item, if we define random variable $X$ to be the number of small buyers that cross the price, and let $\lambda = \E[X]$ to be the expected number of small buyers that cross the price, then it is easy to see that $X \sim \Binom(n, \lambda/n)$.

We can now express the revenue of this static pricing as $\E[\min\{X, k\}]$ (this is the expected number of small buyers that cross the threshold and obtain price $p$), while the utility is $U_k \cdot \pr[X < k]$ (as whenever we do not sell all copies to small buyers, we obtain utility of $U_k$ on expectation from the large buyer). Therefore, it is easy to see that this static price gets a social welfare of $\E[\min\{X, k\}] + U_k \cdot \pr[X < k]$, ignoring lower order terms. If we let $\alpha = \frac{U_k}{k + U_k} \in [0, 1)$, then the competitive ratio that we can achieve if we fix $U_k$ and $\lambda$ is
\[\frac{\E[\min\{X, k\}] + U_k \cdot \pr[X < k]}{k + U_k} = \alpha \cdot \frac{\E[\min\{\Binom(n, \lambda/n), k\}]}{k} + (1 - \alpha) \cdot \pr[\Binom(n, \lambda/n) < k]\]
and therefore the best competitive ratio against the worst case instance (where an adversary chooses $U_k$ and hence $\alpha$, and the mechanism chooses the best $p$ and hence $\lambda$ against this $\alpha$) is exactly
\[\min_{\alpha \in [0, 1]} \max_{\lambda \in [0, n]} \alpha \cdot \frac{\E[\min\{\Binom(n, \lambda/n), k\}]}{k} + (1 - \alpha) \cdot \pr[\Binom(n, \lambda/n) < k]\]

An application of Sion's minimax theorem, as proven in~\cite{jiang2022tightness}, shows that the min and max of the quantity above can be swapped, giving us 
\begin{align*}
    &\max_{\lambda \in [0, n]} \min_{\alpha \in [0, 1]} \alpha \cdot \frac{\E[\min\{\Binom(n, \lambda/n), k\}]}{k} + (1 - \alpha) \cdot \pr[\Binom(n, \lambda/n) < k] \\
    &= \max_{\lambda \in [0, n]} \min\left\{\frac{\E[\min\{\Binom(n, \lambda/n), k\}]}{k},\pr[\Binom(n, \lambda/n) < k]\right\}.
\end{align*} One can observe that $\lambda$ cannot be too large (say, more than $5k$), and with bounded $\lambda$ and $n$ going to infinity, the distribution $\Binom(n, \lambda / n)$ converges to $\Pois(\lambda)$, so the revenue term on the left converges to $\mu_k(\lambda)$, while the utility term on the right converges to $\delta_k(\lambda)$. Finally, since $\mu_k(\lambda)$ is increasing while $\delta_k(\lambda)$ is decreasing in $\lambda$, the maximizer is when both functions are the same, giving us the result.

\subsection{Designing an instance for two items and unit-demand buyers}\label{sec:hard-instance-details}
Inspired by the instance for single-item, we design our instance for two items and unit-demand buyers as follows. Let's fix the multiplicity $k$ for both items (i.e. both items have exactly $k$ copies). Consider $n + 2$ buyers, where the first $n$ of them (called \emph{small buyers}) have valuation profiles that are drawn i.i.d. from the following distribution:
\begin{itemize}
    \item Draw $x \sim \Unif[0, \epsilon]$.
    \item Let the buyer be unit-demand over item 1 and item 2, and their value for item 1 and 2 be $1 + x$ and $1 + (1 + \epsilon)x$ respectively.
\end{itemize}
Furthermore, let there be one buyer who is only interested in item 1, and their value for item 1 is $\frac{U_k}{\epsilon}$ with probability $\epsilon$, and $0$ otherwise; similarly, let there be one buyer who is interested in item 2 with the same distribution. Here, $U_k$ is going to be defined later, and we let $n \to \infty$ while $\epsilon \to 0$ such that $n \epsilon \to 0$.

We observe that the optimal static pricing $(p_1, p_2)$ for any of such instance must satisfy $p_1, p_2 \le 1 + (1 + \epsilon)\epsilon$; if any item price is above this quantity, we can simply let the price of that item be exactly $1 + (1 + \epsilon)\epsilon$ without harming the social welfare of the pricing scheme. Furthermore, for a static pricing $(p_1, p_2)$, we can define the following quantities. Here, we neglect cases where buyers have zero utility for an item, or equal utility for both items, since these events happen with $0$ probability; we also neglect small buyers with negative utility for both items.
\begin{enumerate}
    \item Let random variable $X_{(1)}$ be the number of small buyers that has positive utility for item 1, but negative utility for item 2. We call these buyers \emph{type $(1)$ buyers}.
    \item Let random variable $X_{(2)}$ be the number of small buyers that has positive utility for item 2, but negative utility for item 1. We call these buyers \emph{type $(2)$ buyers}.
    \item Let random variable $X_{(1, 2)}$ be the number of small buyers that has positive utility for both items, but has strictly greater utility for item 1. We call these buyers \emph{type $(1, 2)$ buyers}.
    \item Let random variable $X_{(2, 1)}$ be the number of small buyers that has positive utility for both items, but has strictly greater utility for item 2. We call these buyers \emph{type $(2, 1)$ buyers}.
\end{enumerate}

Let us also define $\lambda_{(1)}, \lambda_{(2)}, \lambda_{(1, 2)}, \lambda_{(2, 1)}$ to be the expected number of buyers of each type. Observe that for every type $t$, we have that $X_{t} \sim \Binom(n, \lambda_t/n)$, since the indicator variable of each small buyer being in this type is an independent Bernoulli with probability $\lambda_t / n$. Furthermore, this works for sum of buyers of many types; for example, we have $X_{(1)} + X_{(1, 2)} \sim \Binom(n, (\lambda_{(1)} + \lambda_{(1, 2)}) / n)$.

We first prove the following structural property of the possible cases for $\lambda$'s.

\begin{restatable}{lemma}{pricesThreeCases}
\label{lem:prices-3-cases}
For any static pricing $(p_1, p_2)$ upon the previous instance, $\lambda$ must satisfy at least one of the three following conditions.
\begin{itemize}
    \item $\lambda_{(1)} = \lambda_{(1, 2)} = 0$.
    \item $\lambda_{(2)} = \lambda_{(2, 1)} = 0$.
    \item $\lambda_{(1)} \le n \epsilon$ and $\lambda_{(2)} \le 2n \epsilon$, which means $\lambda_{(1)} \to 0$ and $\lambda_{(2)} \to 0$ when $n \epsilon \to 0$.
\end{itemize}
\end{restatable}

This structural property is what sets this example apart from the gap example described above for the single item case. At a high level, if we duplicate the single item case into two items, then there are only buyers of type $(1)$ and $(2)$ that do not impose extra demand on each other's item. However, in the current example, we either have the appearance of buyers of type $(1)$ and type $(1, 2)$ (or type $(2)$ and type $(2, 1)$) that impose too much demand on one item compared to the other, or type $(1, 2)$ and type $(2, 1)$ buyers that impose extra demand on each other's favorite items.

Now, we mirror our analysis to that of the single-item case described in the previous section. We first fix $U_k$ and $n$, while still letting $\epsilon \to 0$ such that $n \epsilon \to 0$. Ignoring lower-order terms, the prophet gets $2k + 2U_k$ on expectation. Let us define the following quantity.
\begin{itemize}
    \item $\mu_k'(\lambda_{(1)}, \lambda_{(1, 2)})$ is $\frac{1}{2k}$ times of the expected number of copies from both items that are sold to small buyers against the adversarial order, when the static prices give $\lambda_{(1)}$ buyers of type $(1)$ and $\lambda_{(1, 2)}$ buyers of type $(2)$ in expectation.
    \item $\delta_k'(\lambda_{(1)}, \lambda_{(1, 2)})$ is $\frac{1}{2}$ times the expected number of items whose copies have not been sold out to small buyers sold against the adversarial order, when the static prices give $\lambda_{(1)}$ buyers of type $(1)$ and $\lambda_{(1, 2)}$ buyers of type $(2)$ in expectation.
    \item Similarly, we define $\mu_k''(\lambda_{(1, 2)}, \lambda_{(2, 1)})$ and $\delta_k''(\lambda_{(1, 2)}, \lambda_{(2, 1)})$ for when the static prices give $\lambda_{(1,2)}$ buyers of type $(1,2)$ and $\lambda_{(2,1)}$ buyers of type $(2,1)$ in expectation.
\end{itemize}
These quantities can be defined in terms of the variables $X_{(1)}, X_{(2)}, X_{(1,2)}$, and $X_{(2,1)}$. For example, the expressions for 
$\mu_k''(\lambda_{(1, 2)}, \lambda_{(2, 1)})$ and $\delta_k''(\lambda_{(1, 2)}, \lambda_{(2, 1)})$ are as given below, as for these quantities the order of arrival of the buyers does not matter. Expressions for the other quantities are provided in Appendix~\ref{app:hard-instance}.
\begin{align*}
    \mu_k''(\lambda_{(1, 2)}, \lambda_{(2, 1)}) &= \frac{\E[\min\{X_{(1, 2)} + X_{(2, 1)}, 2k\}]}{2k} \\
    \delta_k''(\lambda_{(1, 2)}, \lambda_{(2, 1)}) &= \frac{\pr[X_{(1, 2)} + X_{(2, 1)} < 2k \cap X_{(1, 2)} < k] + \pr[X_{(1, 2)} + X_{(2, 1)} < 2k \cap X_{(2, 1)} < k]}{2} 
\end{align*}

Once again, ignoring lower order terms, we observe that
\begin{itemize}
    \item When the prices are set such that $\lambda_{(2)} = \lambda_{(2, 1)} = 0$, the revenue we achieve is exactly $2k \cdot \mu_k'(\lambda_{(1)}, \lambda_{(1, 2)})$, while the utility is $2U_k \cdot \delta_k'(\lambda_{(1)}, \lambda_{(1, 2)})$. 
    Therefore, the pricing achieves social welfare of $2k \cdot \mu_k'(\lambda_{(1)}, \lambda_{(1, 2)}) + 2U_k \cdot \delta_k'(\lambda_{(1)}, \lambda_{(1, 2)})$.
    \item The case of $\lambda_{(1)} = \lambda_{(1, 2)} = 0$ similarly obtains a social welfare of $2k \cdot \mu_k'(\lambda_{(2)}, \lambda_{(2, 1)}) + 2U_k \cdot \delta_k'(\lambda_{(2)}, \lambda_{(2, 1)})$.
    \item In the case where $\lambda_{(1)} = \lambda_{(2)} = 0$, we obtain social welfare of $2k \cdot \mu_k''(\lambda_{(1, 2)}, \lambda_{(2, 1)}) + 2U_k \cdot \delta_k''(\lambda_{(1, 2)}, \lambda_{(2, 1)})$.
\end{itemize}
Since the first two cases are symmetric, we may without loss of generality drop one of them from consideration. If we then define $\alpha = \frac{U_k}{U_k + k}$, then the best static pricing against this $U_k$ has competitive ratio of at most
\[\max\left\{\begin{matrix}\displaystyle \max_{\substack{\lambda_{(1)}, \lambda_{(1, 2)} \ge 0 \\ \lambda_{(1)} + \lambda_{(1, 2)} \le n}} \alpha \mu_k'(\lambda_{(1)}, \lambda_{(1, 2)}) + (1 - \alpha) \delta_k'(\lambda_{(1)}, \lambda_{(1, 2)}) \\ \displaystyle \max_{\substack{\lambda_{(1,2)}, \lambda_{(2,1)} \ge 0 \\ \lambda_{(1,2)} + \lambda_{(2,1)} \le n}} \alpha \mu_k''(\lambda_{(1, 2)}, \lambda_{(2, 1)}) + (1 - \alpha) \delta_k''(\lambda_{(1, 2)}, \lambda_{(2, 1)})\end{matrix}\right\} \]
and therefore our worst case approximation for this class of instance is
\begin{equation}\label{eq:full-optimization}
    \min_{\alpha \in [0, 1]} \max\left\{\begin{matrix}\displaystyle \max_{\substack{\lambda_{(1)}, \lambda_{(1, 2)} \ge 0 \\ \lambda_{(1)} + \lambda_{(1, 2)} \le n}} \alpha \mu_k'(\lambda_{(1)}, \lambda_{(1, 2)}) + (1 - \alpha) \delta_k'(\lambda_{(1)}, \lambda_{(1, 2)}) \\ \displaystyle \max_{\substack{\lambda_{(1,2)}, \lambda_{(2,1)} \ge 0 \\ \lambda_{(1,2)} + \lambda_{(2,1)} \le n}} \alpha \mu_k''(\lambda_{(1, 2)}, \lambda_{(2, 1)}) + (1 - \alpha) \delta_k''(\lambda_{(1, 2)}, \lambda_{(2, 1)})\end{matrix}\right\}
\end{equation}

Our next section is dedicated to simplifying this optimization program to the form of~\Cref{thm:worst-case-for-two-items}.

\subsection{Simplifying Expression~\eqref{eq:full-optimization}}\label{sec:hard-instance-simplify}

We first show that the lower optimization problem is, in fact, single-dimensional.

\begin{restatable}{lemma}{lowerOptTwoLambdasEqual}
\label{lem:lower-opt-two-lambdas-equal}
For the lower optimization of $\displaystyle \max_{\substack{\lambda_{(1,2)}, \lambda_{(2,1)} \ge 0 \\ \lambda_{(1,2)} + \lambda_{(2,1)} \le n}} \alpha \mu_k''(\lambda_{(1, 2)}, \lambda_{(2, 1)}) + (1 - \alpha) \delta_k''(\lambda_{(1, 2)}, \lambda_{(2, 1)})$, at optimum, we have $\lambda_{(1, 2)} = \lambda_{(2, 1)}$.
\end{restatable}

This means that if we define
\[
    \hat{\mu}_k(\lambda) = \frac{\E[\min\{\Binom(n, \lambda/n), 2k\}]}{2k}
\]
and
\[
    \hat{\delta}_k(\lambda) = \sum_{i=0}^{2k - 1} \pr[\Binom(n, \lambda/n) = i] \cdot \pr[\Binom(i, 1/2) < k]
\]
then the lower optimization becomes exactly $\max_{\lambda \in [0, n]} \alpha \hat{\mu}_k(\lambda) + (1 - \alpha)\hat{\delta}_k(\lambda)$. $\lambda$ here should be thought of as $\lambda_{(1, 2)} + \lambda_{(2, 1)}$. The term $\pr[\Binom(i, 1/2) < k]$ appears since conditioned on $X_{(1,2)} + X_{(2, 1)} = i$, the distribution of $X_{(1, 2)}$ and $X_{(2, 1)}$ is exactly $\Binom(i, 1/2)$.

Let $f(\alpha, \lambda) = \alpha \hat{\mu}_k(\lambda) + (1 - \alpha)\hat{\delta}_k(\lambda)$. We show that if we ignore the upper optimization, then there is a minimax result for the truncated optimization program.

\begin{restatable}{lemma}{lowerOptMinimax}
\label{lem:lower-opt-minimax}
$\min_{\alpha \in [0, 1]} \max_{\lambda \in [0, n]} f(\alpha, \lambda) = \max_{\lambda \in [0, n]} \min_{\alpha \in [0, 1]} f(\alpha, \lambda).$
\end{restatable}

We achieve this theorem by applying Sion's minimax theorem to $f$, which entails showing that for any fixed $\alpha \in [0, 1]$, the function $f(\alpha, \lambda)$ is unimodal in $\lambda$.
Finally, we show that under the regime $\alpha \in \left[\frac{1}{2}, 1\right]$, we can ignore the upper optimization entirely.

\begin{restatable}{lemma}{alphaOneHalfUpperIgnore}
\label{lem:alpha-one-half-upper-ignore}
When $\alpha \in \left[\frac{1}{2}, 1\right]$, we have 
\begin{align*}
    &\displaystyle \max_{\substack{\lambda_{(1)}, \lambda_{(1, 2)} \ge 0 \\ \lambda_{(1)} + \lambda_{(1, 2)} \le n}} \alpha \mu_k'(\lambda_{(1)}, \lambda_{(1, 2)}) + (1 - \alpha) \delta_k'(\lambda_{(1)}, \lambda_{(1, 2)}) \\
    & \le \displaystyle \max_{\substack{\lambda_{(1,2)}, \lambda_{(2,1)} \ge 0 \\ \lambda_{(1,2)} + \lambda_{(2,1)} \le n}} \alpha \mu_k''(\lambda_{(1, 2)}, \lambda_{(2, 1)}) + (1 - \alpha) \delta_k''(\lambda_{(1, 2)}, \lambda_{(2, 1)}).
\end{align*}
\end{restatable}

Let us figure out how these components combine together. First, we can upper bound Expression~\eqref{eq:full-optimization} with the same program, with the range $\alpha$ reduced to be in $\left[\frac{1}{2}, 1\right]$.

\[
    \text{Expression~\eqref{eq:full-optimization}} \le 
    \min_{\alpha \in \left[\frac{1}{2}, 1\right]} \max\left\{\begin{matrix}\displaystyle \max_{\substack{\lambda_{(1)}, \lambda_{(1, 2)} \ge 0 \\ \lambda_{(1)} + \lambda_{(1, 2)} \le n}} \alpha \mu_k'(\lambda_{(1)}, \lambda_{(1, 2)}) + (1 - \alpha) \delta_k'(\lambda_{(1)}, \lambda_{(1, 2)}) \\ \displaystyle \max_{\substack{\lambda_{(1,2)}, \lambda_{(2,1)} \ge 0 \\ \lambda_{(1,2)} + \lambda_{(2,1)} \le n}} \alpha \mu_k''(\lambda_{(1, 2)}, \lambda_{(2, 1)}) + (1 - \alpha) \delta_k''(\lambda_{(1, 2)}, \lambda_{(2, 1)})\end{matrix}\right\}.
\]

By~\Cref{lem:alpha-one-half-upper-ignore}, we can ignore the upper optimization of the RHS; by~\Cref{lem:lower-opt-two-lambdas-equal}, we can rewrite the lower optimization via a single-parameter optimization. In particular,
$  \text{Expression~\eqref{eq:full-optimization}} \le 
    \min_{\alpha \in \left[\frac{1}{2}, 1\right]} \max_{\lambda \in [0, n]} f(\alpha, \lambda).
$
Now consider the following quantities.
\begin{itemize}
    \item Let $\alpha^*$ and $\lambda^*$ be the solution of \vspace{-2mm}
    \[\min_{\alpha \in [0, 1]} \max_{\lambda \in [0, n]} f(\alpha, \lambda) = \max_{\lambda \in [0, n]} \min_{\alpha \in [0, 1]} f(\alpha, \lambda) = \max_{\lambda \in [0, n]} \min\{\hat{\mu}_k(\lambda), \hat{\delta}_k(\lambda)\},\]
    where the first equation is~\Cref{lem:lower-opt-minimax}. Note that since $\mu$ and $\delta$ are increasing and decreasing in $\lambda$ respectively, $\lambda^*$ must satisfies $\hat{\mu}_k(\lambda) = \hat{\delta}_k(\lambda)$. Furthermore, by optimality condition, we must have $\frac{d}{d \lambda} f(\alpha^*, \lambda^*) = \alpha^* \left(\frac{d}{d \lambda}\hat{\mu}_k(\lambda^*)\right) + (1 - \alpha^*) \left(\frac{d }{d \lambda}\hat{\delta}_k(\lambda^*)\right) = 0$.
    \item Let $\lambda'$ be the solution of $\max_{\lambda \in [0, n]} f\left(\frac{1}{2}, \lambda\right)$. By optimality condition, we must have $\frac{d}{d \lambda} f(\frac{1}{2}, \lambda^*) = \frac{1}{2} \left(\frac{d}{d \lambda}\hat{\mu}_k(\lambda')\right) + (1 - \frac{1}{2}) \left(\frac{d }{d \lambda}\hat{\delta}_k(\lambda')\right) = 0$, or $\left(\frac{d}{d \lambda}\hat{\mu}_k(\lambda')\right) + \left(\frac{d }{d \lambda}\hat{\delta}_k(\lambda')\right) = 0$. Note that as $f(\frac{1}{2}, \lambda)$ is unimodal in $\lambda$, this value is unique.
\end{itemize}
We argue that the RHS of the above equation is at most $\frac{1}{2} \left(\hat{\mu}_k(\max\{\lambda^*, \lambda'\}) + \hat{\delta}_k(\max\{\lambda^*, \lambda'\})\right)$.
\begin{itemize}
\item When $\alpha^* \ge \frac{1}{2}$, we note that the RHS is exactly $f(\alpha^*, \lambda^*) = \frac{1}{2} \left(\hat{\mu}_k(\lambda^*) + \hat{\delta}_k(\lambda^*)\right)$, where the last equality is because $\hat{\mu}_k(\lambda) = \hat{\delta}_k(\lambda)$. Furthermore, we also have $\alpha^* \left(\frac{d}{d \lambda}\hat{\mu}_k(\lambda^*)\right) + (1 - \alpha^*) \left(\frac{d }{d \lambda}\hat{\delta}_k(\lambda^*)\right) = 0$, which means $\left(\frac{d}{d \lambda}\hat{\mu}_k(\lambda^*)\right) + \left(\frac{d }{d \lambda}\hat{\delta}_k(\lambda^*)\right) \le 0$. Since $f(\frac{1}{2}, \lambda)$ is unimodal in $\lambda$ and $\lambda'$ is its peak, we have $\lambda^* \ge \lambda'$.
\item When $\alpha^* < \frac{1}{2}$, RHS is at most $\max_{\lambda \in [0, n]} f\left(\frac{1}{2}, \lambda\right) = \frac{1}{2} \left(\hat{\mu}_k(\lambda')+ \hat{\delta}_k(\lambda')\right)$. From the same argument as above, we have $\left(\frac{d}{d \lambda}\hat{\mu}_k(\lambda^*)\right) + \left(\frac{d }{d \lambda}\hat{\delta}_k(\lambda^*)\right) > 0$, so $\lambda' > \lambda^*$.
\end{itemize}

Therefore, we showed that $\text{Expression~\eqref{eq:full-optimization}} \le \frac{1}{2} \left(\hat{\mu}_k(\max\{\lambda^*, \lambda'\}) + \hat{\delta}_k(\max\{\lambda^*, \lambda'\})\right)$, and with taking $n \to \infty$ so that $\Binom(n, \lambda / n) \to \Pois(\lambda)$, we obtain the exact expression as that of~\Cref{thm:worst-case-for-two-items}.

\newpage
\bibliographystyle{ACM-Reference-Format}
\bibliography{refs}

\newpage
\appendix
\section{The competitive ratio of general online allocation}
\label{sec:online}

In this section we will prove that in the multi-unit combinatorial setting general online allocation can achieve a competitive ratio of $\compratio(\onl, k) = 1-1/\sqrt{k+3}$, matching the performance of online allocation algorithms for the single-item special case.
In other words, we show that the combinatorial prophet inequality with XOS valuation functions is no harder than the single-item inequality.

\begin{theorem}
\label{thm:xos-to-single-reduction}
    Given any instance $(\D, K)$ and a feasible solution $\allocs=\{\alloci[i,v_i,S],\forall i\in [n],v_i\}$ to the linear program~\ref{Prog:EAOPT}, there exists an online allocation mechanism such that, the expected welfare of the mechanism is at least $(1 - 1/\sqrt{k+3}) \cdot \sum_i \sum_{S\subseteq[m]}\E_{v_i\sim\D_i}\left[v_i(S)\,\alloci[i,v_i,S]\right]$, where $k = \min_j k_j$.
\end{theorem}

Our main idea of proving \Cref{thm:xos-to-single-reduction} is to reduce the combinatorial prophet inequality to single-item prophet inequality, and apply the following result of \citet{alaei2014bayesian} as a black box.


\begin{lemma}[\cite{alaei2014bayesian}]
\label{lma:alaei}
    Given any \textbf{single-item} prophet inequality instance $(\D^{single}, k)$ and a feasible solution $\allocs=\{\alloci[i,v_i],\forall i\in [n],v_i\}$, there exists an online allocation mechanism such that, the competitive ratio against ex-ante optimum is at least $\gamma_k \geq 1 - 1/\sqrt{k+3}$. 
\end{lemma}

With \Cref{lma:alaei}, we claim that the following \Cref{alg:xos-to-si} is the desired algorithm for \Cref{thm:xos-to-single-reduction}. The algorithm tries to mimic the given ex-ante solution $\allocs$, but it ``rounds'' this solution with the help of $m$ independent prophet inequalities, one for each item. For each prophet inequality, the algorithm provides as input the supporting value of the corresponding item at an appropriately chosen set, distributed per the intended allocation $\allocs$. We now describe the details.  

\begin{algorithm}[tbh]
\caption{Algorithm for Combinatorial Prophet Inequality}
\label{alg:xos-to-si}
\KwIn{Combinatorial prophet inequality instance $(\D, K)$}
Initiate $m$ single-item prophet inequality instances. Let $\si_j$ be the $j$-th instance with $k_j$ being the budget. \\
\For{$i = 1 \to n$}
{
    Realize $v_i \sim \D_i$. \\
    Draw subset $S$ from distribution $\left\{\alloci[i, v_i, S]\right\}$. \\
    Let $a^{v_i, S}\in \mathcal{A}^{v_i}$ be the supporting additive function for $v_i(S)$, i.e., $\sum_{j \in S} a^{v_i, S}_j = v_i(S) $. \\
    Initiate $S_i = \emptyset$. \\
    \For{$j \in S$}
    {
        Send value $a^{v_i, S}_j$ to $\si_j$. \\
        If $\si_j$ accepts item $j$, then add $j$ into $S_i$.
    }
    Allocate $S_i$ to buyer $i$, and collect $v_i(S_i)$. 
}
\end{algorithm}

\begin{proof}[Proof of \Cref{thm:xos-to-single-reduction}]
    We first analyze the single-item prophet inequality instance we send to $\si_j$. When buyer $i$ arrives, we send value $a^{v_i,S}_j$ to $\si_j$ with probability $\alloci[i, v_i, S] \cdot \pr[v_i \text{ is realized}]$, and $0$ otherwise, which corresponds to the case that the subset $S$ drawn in \Cref{alg:xos-to-si} does not contain $j$. This gives the value distribution for $\si_j$.

    Now, we consider the feasible ex-ante benchmark we work on for $\si_j$. Note that $\{\alloci[i, v_i, S] \cdot \pr[v_i \text{ is realized}]\}$ is a feasible ex-ante solution for $\si_j$. Intuitively speaking, this solution suggests that we should accept any non-zero value sent to $\si_j$,  since in expectation there are at most $k_j$ non-zero values sent to $\si_j$, as exhibited by the following ex-ante constraint
    \[
    \sum_{i\in [n]} \sum_{S\subseteq[m]: S\ni j} \E_{v_i\sim\D_i}\left[\alloci[i,v_i,S]\right] ~\le~ k_j
    \]
    and the assumption that $\{\alloci[i,v_i],\forall i\in [n],v_i\}$ is a feasible solution. We define
    \[
    \sw_j ~:=~ \sum_{i \in [n]} \sum_{v_i \in \operatorname{support}(\D_i)} \sum_{S \subseteq [m]: S \ni j} a^{v_i,S}_j \cdot \alloci[i, v_i, S] \cdot \pr[v_i \text{ is realized}]
    \]
    to be the benchmark for $\si_j$. Then, \Cref{lma:alaei} guarantees that instance $\si_j$ gets at least $ \gamma_k  \cdot \sw_j$ with $\gamma_k \geq 1 - 1/\sqrt{k+3}$ (recall that $k$ is the minimum value among all $k_j$). Summing the gains from all $\si_j$ together, the total gain is at least
    \begin{align*}
        \gamma_k \cdot \sw_j ~&=~ \gamma_k \cdot \sum_{i \in [n]} \sum_{v_i \in \operatorname{support}(\D_i)}  \pr[v_i \text{ is realized}] \cdot \sum_{S \subseteq [m]}  \alloci[i, v_i, S] \cdot \sum_{j \in S} a^{v_i, S}_j \\
        ~&=~\gamma_k \cdot \sum_{i \in [n]} \sum_{v_i \in \operatorname{support}(\D_i)}  \pr[v_i \text{ is realized}] \cdot \sum_{S \subseteq [m]}  \alloci[i, v_i, S] \cdot v_i(S) \\
        ~&=~ \gamma_k \cdot \sum_{i \in [n]} \sum_{S\subseteq[m]}\E_{v_i\sim\D_i}\left[v_i(S)\,\alloci[i,v_i,S]\right] \\
        ~&\geq~ \left(1 - \frac{1}{\sqrt{k+3}}\right)\cdot \sum_{i \in [n]} \sum_{S\subseteq[m]}\E_{v_i\sim\D_i}\left[v_i(S)\,\alloci[i,v_i,S]\right].
    \end{align*}
    Observe that the last line is the objective value of solution $\{\alloci[i,v_i,S]\}$. To prove \Cref{thm:xos-to-single-reduction}, it remains to show that in expectation the summation of $v_i(S_i)$ in \Cref{alg:xos-to-si} is at least the total gain from all instances $\si_j$. Note that for any realization of $v_i$, the property of XOS function guarantees
    \[
    v_i(S_i) ~\geq~ \sum_{j \in S_i} a^{v_i, S}_j.
    \]
    Taking the expectations on both sides and summing over $i \in [n]$ finishes the proof.
\end{proof}


\section{A tight instance for supply based pricing}
\label{sec:supp-tight}

In this section we prove Theorem~\ref{thm:supply-tightness}, which demonstrates that the competitive ratio in Corollary~\ref{cor:opt} is tight if the mechanism designer only knows the value of the optimal ex-ante solution.

\begin{theorem}[Restatement of Theorem~\ref{thm:supply-tightness}]
    There exists a family $\F$ of single-item $k$-unit instances such that:
    \begin{itemize}
        \item For every instance $(\D,k)\in \F$, the value of \eqref{Prog:EAOPT} is $2k$.
        \item For every supply-based pricing vector $p=(p_1, \cdots, p_k)$, there exists an instance $(\D_p, k)\in \F$ with $\welfare(p, \D_p, k) \le (1-(k/k+1)^k)\cdot 2k$.
    \end{itemize}
    Consequently, a supply based pricing mechanism that knows the value of \eqref{Prog:EAOPT} for a given instance but not the value distribution  cannot obtain a competitive ratio better than $1-(k/k+1)^k$.
\end{theorem}


\begin{proof}
    Our hard instance only contains one item. For ease of notation, we use $v_i$ to denote buyer $i$'s valuation of the item. We'll construct a family of instances with $\eaopt(\D, k)=2k$. Each instance $(\D_p, k)$ is parameterized by a static supply based pricing $p$ such that the pricing receives a total social welfare of at most $(1-(\frac{k}{k+1})^k)\cdot 2k$ on this instance.
    

    Denote the supply-based static pricing scheme as $p=(p_1,p_2,\cdots,p_k)$, where $p_c$ is the price of the $c$-th copy of the item. Note that one of the below inequalities must be true:
    \begin{itemize}
        \item $1-k\frac{p_1}{2k}\le 1-(\frac{k}{k+1})^k$
        \item $1-k\frac{p_2}{2k}+\frac{p_1}{2k}\le 1-(\frac{k}{k+1})^k$
        \item $1-k\frac{p_3}{2k}+\frac{p_1+p_2}{2k}\le 1-(\frac{k}{k+1})^k$
        \item $\cdots$
        \item $1-k\frac{p_k}{2k}+\frac{p_1+p_2+\cdots+p_{k-1}}{2k}\le 1-(\frac{k}{k+1})^k$
        \item $\frac{p_1+p_2+\cdots+p_{k}}{2k}\le 1-(\frac{k}{k+1})^k$
    \end{itemize}
This can be seen by taking a suitable linear combination of all $k$ quantities. Let's assume the $c$-th inequality is true without loss of generality.

Now let's construct our instance $(\D_p, k)$. If $c>1$, we let the first $c-1$ buyers to have value $p_i$ for buyer $i$. If $c<k-1$, we let the following $k$ buyers have value $p_c-\delta$. Finally, we let the last buyer have value $\frac{2k-kp_{c}}{\epsilon}$ with probability $\epsilon$ and value $0$ with probability $1-\epsilon$. We let $\epsilon \rightarrow 0$ and $\delta \rightarrow 0$.

We can observe that $\eaopt(\D_p, k)=2k$ since we'll allocate item to the last buyer with $\epsilon$ probability and the $k$ buyers with value $p_c-\delta$ with $1-\frac{\epsilon}{k}$ probability. However, our supply-based static pricing scheme only allocates to the first $c-1$ buyers and to the last buyer with $\epsilon$ probability, achieving a total social welfare of $p_1+p_2+\cdots+p_{c-1}+2k-kp_c$. According to our assumption that the $c$-th inequality is true, this quantity is at most $(1-\frac{k}{k+1})^k\cdot 2k$.
\end{proof}
 
\section{Multi unit demands}\label{sec:multi_unit}

In this section, we introduce a more general model, in which each buyer may demand more than one unit of each item. We assume that $\ell$ is defined such that no buyer demands more than $k_j/\ell$ units of item $j$ for $j\in [m]$. In other words, $\ell$ is the maximum fraction of the supply of an item a buyer can purchase. We show that we can reduce this setting to the combinatorial setting with single-unit-demand per item where the effective supply of each item is $\ell$.
Applying the supply-based pricing scheme in Section~\ref{sec:pricing}, we then obtain a $(1-(\frac{\ell}{\ell+1})^\ell)$ competitive ratio against the optimal social welfare.

We first formally introduce the setting. For a multiset $S$ of items, let $\countSj$ denote the number of copies of item $j$ in the set.

{\bf Multi unit extension of fractionally subadditive (XOS) values:} we say that a valuation function $v$ over multisets of items is fractionally subadditive if there exists a set $\mathcal{A}^v$ of $m$-dimensional vectors $a$ where each component $a_j$ is a concave function over the positive integers, 
such that $v(S) = \max_{a \in \mathcal{A}^v} \sum_{j \in [m]} a_j(\countSj)$ for all multisets $S$ over the item set $[m]$. We say that the valuation function has maximum demand $d$ for item $j$ if all functions $a_j$ in the set $\mathcal{A}^v$ are constant on arguments $\ge d$. In other words, the buyer receives no extra value from extra copies of item $j$ beyond the first $d$.

We denote an instance of the multi-unit demand setting as $(\D, K, \ell)$ where $\ell$ is defined such that the maximum demand for item $j\in [m]$ of any value function in the support of $\D$ is at most $k_j/\ell$.

The definitions of the hindsight optimum and the prophet benchmark extend trivially to this setting.

\begin{align*}
    \hopt(\vals, K,\ell):= \max \sum_i v_i(\alloci) & \text{ subject to }\\
    \sum_{i\in [n]} x_{ij} & \le k_j & \forall j\in [m]\\
    x_{ij} &\in \left\{0,1,\cdots,\Big\lfloor \frac{k_j}{\ell}\Big\rfloor\right\} & \forall i\in [n], j\in [m]
\end{align*}

\begin{align*}
    \eaopt(\D, K,\ell):= \max \sum_i \sum_{S\subseteq[m]}\E_{v_i\sim\D_i}\left[v_i(S)\,\alloci[i,v_i,S]\right] & \text{ subject to } \\
    \sum_{i\in [n]} \sum_{S\subseteq[m]: S\ni j} \E_{v_i\sim\D_i}\left[\alloci[i,v_i,S]\right] & \le k_j & \forall j\in [m]\\
    \sum_{S\subseteq [m]}\alloci[i,v_i,S] &\le \Big\lfloor \frac{k_j}{\ell}\Big\rfloor & \forall i\in[n], v_i\in \operatorname{support}(\D_i)\\
    \alloci[i,v_i,S] &\in [0,1] & \forall i\in [n], S\subseteq [m]
\end{align*}

\begin{theorem}
    Given any instance $(\D, K, \ell)$ in the multi unit demand setting, there exists a supply-based static pricing mechanism such that, the expected welfare of the mechanism is at least $(1-(\frac{\ell}{\ell+1})^\ell)$ times $\eaopt(\D,K,\ell)$.
\end{theorem}

\begin{proof}
    We'll start the proof by reducing our instance into a simpler model. Define $c_j=\lfloor \frac{k_j}{l}\rfloor$ and divide the $k_j$ units of item $j$ into $c_j$ bins almost equally, where each bin contains $\ell$ or $\ell+1$ units of item $j$. Create a new item type for each bin (units from the same bin has the same item type, but units from different bins are treated as different types of items).

    Now we apply the supply-based pricing scheme in Section~\ref{sec:pricing}. Since any buyer in the original model wants no more than $c_j$ units of item $j$, we can simulate the allocation by allocating one unit from each item type in the transformed model (which corresponding to item $j$ in the original model). Thus by Corollary~\ref{cor:opt}, we obtain a $(1-(\frac{\ell}{\ell+1})^\ell)$ approximation to $\eaopt(\D,K,\ell)$.
\end{proof}

\section{Deferred proofs from \Cref{sec:dynamic}}
\label{sec:dyn-proofs}

\begin{proof}[Proof of \Cref{lma:dyn-prices}]
    Note that the linear program in Line 5 is feasible. 
    Consider to take the dual of the program, we have 
    \begin{align*}
        \qquad \min  \sum_{j \in R_i} \rho_j \cdot z_{i, j} +  \sum_{v_i \in \operatorname{support}(\D_i)}  \psi_{v_i} \cdot &\pr\left[v_i \text{ is realized}\right] \text{ subject to } \label{Prog:dual} \\
         \sum_{j \in S} \rho_j + \psi_{v_i} & \ge v_i(S) & \forall S \subseteq R_i, v_i \in \operatorname{support}(\D_i)\\
        \rho_j &\geq 0 & \forall j \in R_i \\
        \psi_{v_i} &\geq 0 & \forall v_i \in \operatorname{support}(\D_i)
    \end{align*}
Let $\{\rho^*_j\}, \{\psi^*_{v_i}\}$ be the optimal solution of the dual program. Since the linear program is feasible, the complementary slackness guarantees that for every $v_i, S$ such that $\ysalloci[v_i, S] > 0$, there must be $v_i(S) - \sum_{j \in S} \rho_j = \psi_{v_i}$. As a corollary, we have 
\[
v_i(S') - \sum_{j \in S'} \rho_j ~<~ \psi^*_{v_i} ~=~v_i(S) - \sum_{j \in S} \rho_j
\]
for all other $S'$ that satisfies $\ysalloci[v_i, S'] = 0$.

Now consider to set $p^{(i)}_j = \rho^*_j$ for every $j \in R_i$. If valuation $v_i$ is realized, the above inequality guarantees that the set $\arg \max_{S} \left(v_i(S) - \sum_{j \in S} p^{(i)}_j \right)$ can only contain $S$ that satisfies $\ysalloci[v_i, S] > 0$. Now consider to allocate subset $S$ to buyer $i$ with probability $\ysalloci[v_i, S] / \pr\left[v_i \text{ is realized}\right]$, we have
\begin{align*}
    &~\pr\left[v_i \text{ is realized} ~\land ~ S = \arg \max_{S} \left(v_i(S) - \sum_{j \in S} p^{(i)}_j \right)\right] \\
    ~=&~ \pr\left[v_i \text{ is realized}\right] \cdot \frac{\ysalloci[v_i, S]}{\pr\left[v_i \text{ is realized}\right]} ~=~ \ysalloci[v_i, S]. \qedhere
\end{align*}
\end{proof}

\section{Missing proofs and expressions for~\Cref{sec:hard-instance}}\label{app:hard-instance}

\subsection{Missing proofs and expressions for~\Cref{sec:hard-instance-details}}

We list the detailed expressions for $\mu'(\lambda_{(1)}, \lambda'_{(2)})$ and $\delta'(\lambda_{(1)}, \lambda'_{(2)})$ here. Define the random variable $X = X_{(1)} + X_{(1, 2)}$, then
\begin{align*}
\mu'(\lambda_{(1)}, \lambda_{(1, 2)}) &= \pr[X < k] \cdot \frac{\E[X \mid X < k]}{2k} + \pr[X \in [k, 2k) \cap X_{(1, 2)} < k] \cdot \frac{1}{2} \\
&\quad+\pr[X \in [k, 2k) \cap X_{(1, 2)} \ge k] \cdot \frac{\E[X \mid X \in [k, 2k) \cap X_{(1, 2)} \ge k]}{2k} \\
&\quad+\pr[X \ge 2k \cap X_{(1, 2)} < k] \cdot \frac{1}{2} + \pr[X \ge 2k \cap X_{(1, 2)} \ge k]
\end{align*}
and
\begin{align*}
\delta'(\lambda_{(1)}, \lambda_{(1, 2)}) &= \pr[X < k] + \pr[X \in [k, 2k) \cap X_{(1, 2)} < k] \cdot \frac{1}{2} \\
&\quad+\pr[X \in [k, 2k) \cap X_{(1, 2)} \ge k] \cdot \frac{1}{2} +\pr[X \ge 2k \cap X_{(1, 2)} < k] \cdot \frac{1}{2}.
\end{align*}

\begin{proof}
Consider the following $5$ cases.
\begin{enumerate}
    \item When $X < k$, we simply send the small buyers in any order. We sell $X$ copies, and both items remain available at the end, so $\mu'$ gains $\frac{X}{2k}$ while $\delta'$ gains $1$ in this case.
    \item When $X \in [k, 2k)$ and $X_{(1, 2)} < k$, we send buyers of type $(1, 2)$ first, then of type $(1)$. Note that we sell all $k$ copies from the first item, while nothing for the second item, so both $\mu'$ and $\delta'$ gain $\frac{1}{2}$.
    \item When $X \in [k, 2k)$ and $X_{(1, 2)} \ge k$, we send buyers of type $(1)$ first, then of type $(1, 2)$. As $X_{(1)} < k$, we in fact sell $X$ copies here and run out of copies for the first item. Therefore, $\mu'$ gains $\frac{X}{2k}$, while $\delta'$ gains $\frac{1}{2}$.
    \item When $X \ge 2k$ and $X_{(1, 2)} < k$, we let the $(1, 2)$ buyers go first, and then the $(1)$ buyers. This buys all copies of the first item while not touching the second item, so both $\mu'$ and $\delta'$ gains $\frac{1}{2}$.
    \item When $X \ge 2k$ and $X_{(1, 2)} \ge k$, we let the $(1)$ buyers go first, and then the $(1, 2)$ buyers. In fact, all copies of both items will be sold here, so $\mu'$ gains $1$ while $\delta'$ gains $0$.
\end{enumerate}
\end{proof}

\pricesThreeCases*

\begin{proof}[Proof of~\Cref{lem:prices-3-cases}]
Recall that for a small buyer, their value vector is $v = (1 + x, 1 + (1 + \epsilon)x)$, where $x \sim \Unif[0, \epsilon]$. Therefore, with a static price vector $p = (p_1, p_2)$, their utility vector for each item is $u = (1 + x - p_1, 1 + (1 + \epsilon)x - p_2)$. We consider three cases on $p$.

\paragraph{Case 1: $p_1 \le p_2 - \epsilon^2$} In this case, we argue that no buyers have $u_2 > u_1$, which implies that $\lambda_{(2)} = \lambda_{(2, 1)} = 0$. This is because in order for $u_2 > u_1$, we must have
\begin{align*}
    1 + x(1 + \epsilon) - p_2 > 1 + x - p_1 \Leftrightarrow x \epsilon > p_2 - p_1 \Leftrightarrow \frac{p_2 - p_1}{\epsilon} \ge \epsilon
\end{align*}
which is not possible.

\paragraph{Case 2: $p_1 \ge p_2 + \epsilon^2$} In this case, we argue that no buyers have $u_1 > u_2$, which implies that $\lambda_{(1)} = \lambda_{(1, 2)} = 0$. The proof is exactly the same as the above case.

\paragraph{Case 3: $|p_1 - p_2| \le \epsilon^2$} Let us argue that almost no buyers are of type $(1)$ or type $(2)$. For a buyer to be of type $(1)$, we must have
\begin{align*}
    1 + x - p_1 &> 0 > 1 + x(1 + \epsilon) - p_2 \\
        & x \in \left(p_1 - 1, \frac{p_2 - 1}{1 + \epsilon}\right) \cap [0, \epsilon]
\end{align*}
and observe that the size of the first range is at most $\frac{p_2 - 1}{1 + \epsilon} - (p_1 - 1) \le p_2 - p_1 \le \epsilon^2$, so the probability that $x$ is in this range is at most $\epsilon$. This means that $\lambda_{(1)} \le n \epsilon$

Similarly, we can argue that for a buyer to be of type $(2)$, we must have $x \in \left(\frac{p_2 - 1}{1 + \epsilon}, p_1 - 1\right) \cap [0, \epsilon]$. The size of the first range is at most
\[\frac{(p_1 - 1)(1 + \epsilon) - (p_2 - 1)}{1 + \epsilon} \le \frac{\epsilon^2 + \epsilon((1 + (1 + \epsilon)\epsilon - 1)}{1 + \epsilon} = 2\epsilon^2\]
where we used that $p_1 \le 1 + (1 + \epsilon)\epsilon$. Therefore, the probability that $x$ is in this range is at most $2 \epsilon$, which implies $\lambda_{(2)} \le 2n\epsilon$.
\end{proof}

\subsection{Missing proofs for~\Cref{sec:hard-instance-simplify}}
Before proceeding, we state the following technical lemma.

\begin{lemma}\label{lem:derivative-of-sum-of-binomial}
    For any integers $0 \le a \le b \le n$, real $p \in [0, 1]$, and function $f : \Z_{\ge 0} \to \R$ we have
    \begin{align*}
        &\frac{d}{dp} \sum_{i=a}^b f(i) \binom{n}{i} p^i (1-p)^{n-i} \\
        &= af(a) \binom{n}{a} p^{a - 1} (1 - p)^{n-a} - (n-b) f(b) \binom{n}{b} p^{b} (1 - p)^{n-b-1} \\
        &\quad\quad + \sum_{i=a}^{b-1} \left(f(i+1) - f(i)\right) \binom{n}{i+1} (i+1) p^{i} (1 - p)^{n-i-1}
    \end{align*}
\end{lemma}

\begin{proof}[Proof of~\Cref{lem:derivative-of-sum-of-binomial}]
\begin{align*}
    & \frac{d}{dp} \sum_{i=a}^{b} f(i) \binom{n}{i} p^i (1 - p)^{n-i} \\
    &= \sum_{i=a}^{b} \left[f(i) \binom{n}{i} i p^{i - 1} (1 - p)^{n-i} - f(i) \binom{n}{i}(n-i) p^{i} (1 - p)^{n-i-1}\right] \\
    &= \sum_{i=a-1}^{b-1} f(i+1)\binom{n}{i+1} (i+1) p^{i} (1 - p)^{n-i-1} - \sum_{i=a}^b f(i) \binom{n}{i}(n-i) p^{i} (1 - p)^{n-i-1} \\
    &= f(a) \binom{n}{a} a p^{a - 1} (1 - p)^{n-a} - f(b) \binom{n}{b}(n-b) p^{b} (1 - p)^{n-b-1} \\
    &\quad+ \sum_{i=a}^{b-1} \left(f(i+1) \binom{n}{i+1} (i+1) - f(i) \binom{n}{i}(n-i)\right) p^{i} (1 - p)^{n-i-1} \\
    &= f(a) \binom{n}{a} a p^{a - 1} (1 - p)^{n-a} - f(b) \binom{n}{b}(n-b) p^{b} (1 - p)^{n-b-1} \\
    &\quad+ \sum_{i=a}^{b-1} \left(f(i+1) - f(i)\right) \binom{n}{i+1} (i+1) p^{i} (1 - p)^{n-i-1} \\
\end{align*}
where the last equality is because $\binom{n}{i+1} (i+1) = \binom{n}{i}(n-i)$ for all $i, n$. 
\end{proof}

\alphaOneHalfUpperIgnore*
\begin{proof}[Proof of~\Cref{lem:alpha-one-half-upper-ignore}]
Fix any $\alpha \in \left[\frac{1}{2}, 1\right]$. We will argue that at optimum, $\displaystyle \max_{\substack{\lambda_{(1)}, \lambda_{(1, 2)} \ge 0 \\ \lambda_{(1)} + \lambda_{(1, 2)} \le n}} \alpha \mu_k'(\lambda_{(1)}, \lambda_{(1, 2)}) + (1 - \alpha) \delta_k'(\lambda_{(1)}, \lambda_{(1, 2)})$ must have $\lambda_{(1)} = 0$. Then, the only positive parameter is $\lambda_{(1, 2)}$, which is directly a subcase of the second optimization program, thus proving our statement. For the rest of this section, we will only focus on the first optimization program.

Consider any $(\lambda_{(1)}, \lambda_{(1, 2)}) = (x, y)$ for some $x, y \ge 0$. We will argue that either replacing $(\lambda_{(1)}, \lambda_{(1, 2)})$ with $(0, x + y)$ gives a no-worse objective. Let us call the profiles $(x, y)$ and $(0, x + y)$ as $A$ and $B$ respectively.

Let us couple the draws of $X^A_{(1)}$ and $X^A_{(1, 2)}$ from profile $A$ with draws of $X^B_{(1)}$ and $X^B_{(1, 2)}$ from profile $(0, x + y)$.
\begin{itemize}
    \item When we draw $X^A_{(1)}$ and $X^A_{(1, 2)}$ from the profile $A$, return $X^B_{(1)} = 0, X^B_{(1, 2)} = X^A_{(1)} + X^A_{(1, 2)}$ for the draw from profile $B$.
\end{itemize}

It is easy to verify that this coupling preserves the distribution of draws for profiles $B$.

Via this coupling, we now show that $\alpha \mu'_k + (1 - \alpha) \delta'_k$ is bigger for profile $B$ for all draws of $X^A_{(1)}$ and $X^A_{(1, 2)}$. For ease of presentation, define $X = X^A_{(1)} + X^A_{(1, 2)} = X^B_{(1)} + X^B_{(1, 2)}$. Then we have 5 cases.

\begin{enumerate}
    \item When $X < k$, we simply send the small buyers in any order. We sell $X$ copies. All profiles gain $\frac{X}{2k} \alpha + (1 - \alpha)$ here.
    \item When $X \in [k, 2k)$ and $X^A_{(1, 2)} < k$, for profile $A$, we send buyers of type $(1, 2)$ first, then of type $(1)$. Note that we sell all $k$ copies from the first item, so profile $A$ gains $\frac{1}{2}$. For profile $B$, we sell $X$ copies but item 1 runs out of copies, so it gets $\frac{X \alpha}{2k} + \frac{1 - \alpha}{2} \ge \frac{k \alpha}{2k} + \frac{1 - \alpha}{2} = \frac{1}{2}$.
    \item When $X \in [k, 2k)$ and $X^A_{(1, 2)} \ge k$, for profile $A$, we send buyers of type $(1)$ first, then of type $(1, 2)$. As $X^A_{(1)} < k$, we in fact sell $X$ copies here while running out of copies on item 1, so profile $A$ gains $\frac{S\alpha}{2k} + \frac{1 - \alpha}{2}$. Profile $B$ is unchanged from the case above, so it gains gain $\frac{S \alpha}{2k} + \frac{1 - \alpha}{2}$ as well.
    \item When $X \ge 2k$ and $X^A_{(1, 2)} < k$, in profile $A$, we let the $(1, 2)$ buyers go first, and then the $(1)$ buyers. This buys all copies of item $1$, so profile $A$ gains $\frac{1}{2}$. For profile $B$, since there are $X \ge 2k$ buyers of type $(1, 2)$, they will buy all copies of both items, letting $B$ gain $\alpha \ge \frac{1}{2}$.
    \item When $X \ge 2k$ and $X^A_{(1, 2)} \ge k$, in profile $A$, we let the $(1)$ buyers go first, and then the $(1, 2)$ buyers. In fact, all copies of both items will be sold here, so profile $A$ gains $\alpha$. Profile $B$ is unchanged from above, so it also gains $\alpha$.
\end{enumerate}

This means that one optimal solution of the first optimization program must have $\lambda_{(1)} = 0$, completing our proof.
\end{proof}

\lowerOptTwoLambdasEqual*
\begin{proof}[Proof of~\Cref{lem:lower-opt-two-lambdas-equal}]
Let's fix any $t \in [0, n]$. We will argue that the $(\lambda_{(1, 2)}, \lambda_{(2, 1)})$ pair that maximizes the above optimization subject to the condition $(\lambda_{(1, 2)}, \lambda_{(2, 1)}) = t$ is exactly the pair $(t/2, t/2)$.

First, we note that $X_{(1, 2)} + X_{(2, 1)} \sim \Binom(n, t/n)$. Furthermore, conditioned on $X_{(1, 2)} + X_{(2, 1)} = v$ for any value $v \ge 0$, the conditional distribution of $X_{(1)}$ is exactly $\Binom\left(v, \frac{\lambda_{(1, 2)}}{\lambda_{(1, 2)} + \lambda_{(2, 1)}}\right)$, while the conditional distribution of $X_{(2)}$ is exactly $\Binom\left(v, \frac{\lambda_{(2, 1)}}{\lambda_{(1, 2)} + \lambda_{(2, 1)}}\right)$.

Let's examine the optimization program. First, note that if we fix $t$, then $\mu_k''(\lambda_{(1, 2)}, \lambda_{(2, 1)}) = \frac{\E[\min\{X_{(1, 2)} + X_{(2, 1)}, 2k\}]}{2k}$ is also fixed. On the other hand,
\begin{align*}
    &\delta_k''(\lambda_{(1, 2)}, \lambda_{(2, 1)}) \\
    &= \frac{\pr[X_{(1, 2)} + X_{(2, 1)} < 2k \cap X_{(1, 2)} < k] + \pr[X_{(1, 2)} + X_{(2, 1)} < 2k \cap X_{(2, 1)} < k]}{2} \\
    &=\frac{\sum_{i=0}^{2k - 1} \pr[X_{(1, 2)} + X_{(2, 1)} = i] \cdot \left(\pr[X_{(1, 2)} < k \mid \pr[X_{(1, 2)} + X_{(2, 1)} = i] + \pr[X_{(2, 1)} < k \mid \pr[X_{(1, 2)} + X_{(2, 1)} = i]\right)}{2} \\
    &=\frac{\sum_{i=0}^{2k - 1} \pr[X_{(1, 2)} + X_{(2, 1)} = i] \cdot \left(\pr\left[\Binom\left(i, \frac{\lambda_{(1, 2)}}{\lambda_{(1, 2)} + \lambda_{(2, 1)}}\right) < k\right] + \pr\left[\Binom\left(i, \frac{\lambda_{(2, 1)}}{\lambda_{(1, 2)} + \lambda_{(2, 1)}}\right) < k\right]\right)}{2}
\end{align*}
Define $p = \frac{\lambda_{(1, 2)}}{\lambda_{(1, 2)} + \lambda_{(2, 1)}} \in [0, 1]$, and let's examine the quantity
\begin{align*}
    \pr\left[\Binom\left(i, \frac{\lambda_{(1, 2)}}{\lambda_{(1, 2)} + \lambda_{(2, 1)}}\right) < k\right] + \pr\left[\Binom\left(i, \frac{\lambda_{(2, 1)}}{\lambda_{(1, 2)} + \lambda_{(2, 1)}}\right) < k\right]
\\= \pr\left[\Binom\left(i, p\right) < k\right] + \pr\left[\Binom\left(i, 1 - p\right) < k\right]
\end{align*} for any fixed $i \le 2k - 1$. We further assume $i \ge k$, as with $i < k$ this quantity is exactly $2$. Note that in that case
\begin{align*}
&\pr\left[\Binom\left(i, p\right) < k\right] + \pr\left[\Binom\left(i, 1 - p\right) < k\right] \\
&= \pr\left[\Binom\left(i, p\right) \le k - 1\right] + \pr\left[\Binom\left(i, p\right) \ge i - k + 1\right].
\end{align*}
Observe that when $i = 2k - 1$ exactly, we have $k - 1 = (i - k + 1) - 1$, so $\pr\left[\Binom\left(i, p\right) \le k - 1\right] + \pr\left[\Binom\left(i, p\right) \ge i - k + 1\right] = 1$. When $i \le 2k - 2$, we have $i - k + 1 \le k - 1$, so the above quantity becomes exactly $1 + \pr[\Binom\left(i, p\right) \in [k - 1, i - k + 1]]$.

As we aim to maximize this quantity, let's take its derivative with respect to $p$:
\begin{align*}
    &\frac{d}{dp} \left(1 + \pr[\Binom(i, p) \in [i - k + 1, k - 1]]\right) \\
    &= \binom{i}{i-k + 1} (i-k + 1) p^{i-k} (1 - p)^{k-1} - \binom{i}{k} k p^{k-1} (1 - p)^{i-k} \\
    &= C (p^{i-k}(1-p)^{k-1} - p^{k-1}(1-p)^{i-k}) \\
    &= Cp^{k-1}(1-p)^{k-1}(p^{2k-1-i}-(1-p)^{2k-1-i}).
\end{align*}
where the first equality is due to~\Cref{lem:derivative-of-sum-of-binomial} with $f(i) = 1$, and the second equality is due to defining $C \coloneqq \binom{i}{i-k + 1} (i-k + 1) = \binom{i}{k} k$.

We need to solve this quantity being zero, and there are three roots to this equation: $0, 1,$ and $\frac{1}{2}$. A direct calculation shows that $p=\frac{1}{2}$ is the maximizer here, and hence showing that $\lambda_{(1, 2)} = \lambda_{(2, 1)}$.
\end{proof}

\lowerOptMinimax*
\begin{proof}[Proof of~\Cref{lem:lower-opt-minimax}]
    Let us define $p = \frac{\lambda}{n}$, and overload $f, \hat{\mu}$, and $\hat{\delta}$ to use with $p$ instead of $\lambda$.

    We apply Sion's minimax theorem here. Note that $f$ is linear in $\alpha$, so we only need to prove that $f$ is quasi-concave in $p \in [0, 1]$ for all fixed $\alpha \in [0, 1]$. This entails proving that $f$ is unimodal over $p$. It suffices to show the following conditions for any $\alpha \in [0, 1]$.
\begin{itemize}
    \item If there exist a root $r \in (0, 1)$ of the equation $\frac{d}{dp} f(\alpha, r) = 0$, it must satisfies that $\frac{d}{dp}f(\alpha, p) \ge 0$ for $p \in [0, r]$ and $\frac{d}{dp}f(\alpha, p) \le 0$ for $p \in [r, 1]$.
\end{itemize}

Let us write out $f$ in detail.
\begin{align*}
    f(\alpha, p) &= \alpha \left(\sum_{i=1}^{2k-1} \frac{i}{2k} \binom{n}{i} p^i (1-p)^{n-i} + \sum_{i=2k}^n \binom{n}{i} p^i (1-p)^{n-i}\right) \\
        & \quad + (1 - \alpha) \left(\sum_{i=0}^{k-1} \binom{n}{i} p^i (1-p)^{n-i} + \sum_{i=k}^{2k-1} \binom{n}{i} p^i (1-p)^{n-i} 2^{-i}\sum_{j=0}^{k-1} \binom{i}{j}\right) \\
\end{align*}

We present our calculation result for $\frac{d}{dp} f(\alpha, p)$ here and defer its calculation to later on.

\begin{lemma}\label{lem:df/dp}
    \begin{align*}
    \frac{d}{dp} f(\alpha, p) &= -(1 - \alpha)k\binom{n}{2k} p^{2k - 1} (1 - p)^{n-2k} \\
    &\quad + n\left( \sum_{i=0}^{2k - 1} \frac{\alpha}{2k} \binom{n - 1}{i} p^{i} (1 - p)^{n-i-1}  - \sum_{i=k-1}^{2k - 2} 2^{-i-1} (1 - \alpha)\binom{i}{k - 1} \binom{n - 1}{i} p^{i} (1 - p)^{n-i-1}\right)
    \end{align*}
\end{lemma}

We first observe that when $\alpha = 0$, the above quantity is trivially non-positive for $p$ over $[0, 1]$. When $\alpha > 0$, we have $\frac{d}{dp}f(\alpha, 0) = \frac{\alpha}{2k} (n-1) > 0$; furthermore, $\frac{d}{dp}f(\alpha, 1) = 0$. Therefore, it suffices to prove that $\frac{d}{dp}f(\alpha, p) \ge 0$ has at most one root in $(0, 1)$; if there exists such a unique root $r$, it is trivial to show that $\frac{d}{dp}f(\alpha, p) \ge 0$ for $p \in [0, r]$ and $\frac{d}{dp}f(\alpha, p) \le 0$ for $p \in [r, 1]$.

It suffices to show that $\frac{d^2}{dp^2} f(\alpha, p) = 0$ has at most one root in $(0, 1)$. If there are no such roots, then $\frac{d}{dp} f(\alpha, p)$ decreases for $p$ from $0$ to $1$. If there is exactly one root, then either $\frac{d}{dp} f(\alpha, p)$ increases then decreases from $0$ to $1$, which means it never crosses $0$ as $\frac{d}{dp} f(\alpha, 0) > 0$ and $\frac{d}{dp} f(\alpha, 1) = 0$, or $\frac{d}{dp} f(\alpha, p)$ decreases then increases, which means it crosses $0$ exactly once when it is decreasing from $0$.

We then calculate $\frac{d^2}{dp^2} f(\alpha, p)$. Its detailed calculation can be found later on.

\begin{lemma}\label{lem:df2/dp2}
\begin{align*}
&\frac{d^2}{dp^2} f(\alpha, p) \\
&= -(1 - \alpha)k\binom{n}{2k} p^{2k - 2} (1 - p)^{n-2k-1}((2k-1)(1-p)-(n-2k)p) \\
&\quad -n(n-2k) \frac{\alpha}{2k} \binom{n-1}{2k-1} p^{2k-1} (1 - p)^{n-2k-1} \\
& \quad +n(n-2k+1) 2^{1-2k} (1 - \alpha) \binom{2k-2}{k - 1} \binom{n-1}{2k-2} p^{2k-2} (1 - p)^{n-2k} \\
&\quad - n\sum_{i=k-2}^{2k-3} 2^{-i-2} (1 - \alpha) \binom{i+1}{k-1} (2k - 3 - i) \binom{n-1}{i+1} p^{i} (1 - p)^{n-i-2}.
\end{align*}
\end{lemma}

Consider the coefficient of $p^{2k-2}(1-p)^{n-2k-1}$ from the first three terms, which is
\begin{align*}
    &-(1 - \alpha)k\binom{n}{2k}\left((2k-1)(1-p)-(n-2k)p\right) -n(n-2k) \frac{\alpha}{2k} \binom{n-1}{2k-1} p \\
    &\quad +n(n-2k+1) 2^{1-2k} (1 - \alpha) \binom{2k-2}{k - 1} \binom{n-1}{2k-2}(1 - p) \\
    &= \frac{(1 - \alpha)k \cdot n!}{(2k)!(n-2k)!}\left(np-p-2k+1\right) -\frac{\alpha n!}{(2k)!(n-2k-1)!} p + \frac{n!2^{1-2k} (1 - \alpha)}{((k-1)!)^2(n-2k)!}(1 - p) \\
    &= \frac{n!}{(2k)!(n-2k)!}\left((1 - \alpha)kn - (1 - \alpha)k - \alpha(n-2k) - \frac{2^{1-2k} (1 - \alpha)(2k)!}{((k-1)!)^2}\right)p \\
    & \quad +\frac{n!}{(2k)!(n-2k)!}\left((1 - \alpha)(1-2k)k + \frac{2^{1-2k} (1 - \alpha) (2k)!}{((k-1)!)^2}\right) \\
    &= \binom{n}{2k}\left[\left(3 \alpha k - \alpha k n + k n - k - \alpha n - \frac{2^{1-2k} (1 - \alpha)(2k)!}{((k-1)!)^2}\right)p +\left((1 - \alpha)(1-2k)k + \frac{2^{1-2k} (1 - \alpha) (2k)!}{((k-1)!)^2}\right)\right]
\end{align*}

Define $A = 3 \alpha k - \alpha k n + k n - k - \alpha n - \frac{2^{1-2k} (1 - \alpha)(2k)!}{((k-1)!)^2}$ and $B = (1 - \alpha)(1-2k)k + \frac{2^{1-2k} (1 - \alpha) (2k)!}{((k-1)!)^2}$, then
\begin{align}\label{eq:reduction2}
\frac{d^2}{dp^2} f(\alpha, p) &= p^{2k-2}(1-p)^{n-2k-1} \binom{n}{2k} (Ap + B) \nonumber\\
&\quad - n\sum_{i=k-2}^{2k-3} 2^{-i-2} (1 - \alpha) \binom{i+1}{k-1} (2k - 3 - i) \binom{n-1}{i+1} p^{i} (1 - p)^{n-i-2}.
\end{align}

Let us argue that $B \le 0$. We note that $\frac{2^{1-2k} (2k)!}{((k-1)!)^2} = \frac{2^{1-2k} (2k-2)! (2k-1) 2k}{((k-1)!)^2} = \binom{2k-2}{k-1} 2^{2-2k} k(2k - 1) \le 2^{2k-2} 2^{2-2k} k(2k-1) = k(2k-1)$. Therefore,
\begin{align*}
    B &= (1 - \alpha)(1-2k)k + \frac{2^{1-2k} (1 - \alpha) (2k)!}{((k-1)!)^2} \\
    &\le (1 - \alpha)(1-2k)k + k(2k - 1)(1 - \alpha) = 0
\end{align*}

Finally, if we divide~\Cref{eq:reduction2} by $p^{2k-3}(1-p)^{n-2k-1}$ and rearrange, then solving $\frac{d^2}{dp^2} f(\alpha, p) = 0$ is equivalent to solving
\begin{align*}
\frac{\binom{n}{2k}(A+\frac{B}{p})}{n} =\sum_{i=k-2}^{2k-3} 2^{-i-2} (1 - \alpha) \binom{i+1}{k-1} (2k-3-i) \binom{n-1}{i+1} \left(\frac{1-p}{p}\right)^{2k - 1 - i}.
\end{align*}

Observe that the LHS is non-increasing in $p$ as $B \le 0$, while the RHS is strictly decreasing as every coefficient of $\left(\frac{1-p}{p}\right)^{2k - 1 - i}$ is non-negative, with some being strictly positive. Therefore, our statement is proven.

\end{proof}

\begin{proof}[Proof of~\Cref{lem:df/dp}]
Recall that
\begin{align*}
    f(\alpha, p) &= \alpha \left(\sum_{i=1}^{2k-1} \frac{i}{2k} \binom{n}{i} p^i (1-p)^{n-i} + \sum_{i=2k}^n \binom{n}{i} p^i (1-p)^{n-i}\right) \\
        & \quad + (1 - \alpha) \left(\sum_{i=0}^{k-1} \binom{n}{i} p^i (1-p)^{n-i} + \sum_{i=k}^{2k-1} \binom{n}{i} p^i (1-p)^{n-i} 2^{-i}\sum_{j=0}^{k-1} \binom{i}{j}\right) \\
\end{align*}

Let's apply~\Cref{lem:derivative-of-sum-of-binomial} to the derivative of each of the sum in the formula for $f(\alpha, p)$. For $\alpha \sum_{i=1}^{2k-1} \frac{i}{2k} \binom{n}{i} p^i (1-p)^{n-i}$, with $f(i) = \frac{\alpha i}{2k}$, we have
\begin{align*}
    &\frac{d}{dp} \left(\sum_{i=1}^{2k-1} \frac{i}{2k} \binom{n}{i} p^i (1-p)^{n-i}\right) \\
    &= \frac{\alpha}{2k} n (1 - p)^{n-1} - (n-2k+1) \frac{\alpha(2k-1)}{2k} \binom{n}{2k - 1} p^{2k - 1} (1 - p)^{n-2k} \\
        &\quad + \sum_{i=1}^{2k - 2} \frac{\alpha}{2k} \binom{n}{i+1} (i+1) p^{i} (1 - p)^{n-i-1}.
\end{align*}

For $\alpha \sum_{i=2k}^n \binom{n}{i} p^i (1-p)^{n-i}$, with $f(i) = \alpha$, we have
\begin{align*}
    \frac{d}{dp} \left(\alpha \sum_{i=2k}^n \binom{n}{i} p^i (1-p)^{n-i}\right) = 2k \alpha \binom{n}{2k} p^{2k - 1} (1 - p)^{n-2k}.
\end{align*}

For $(1 - \alpha) \sum_{i=0}^{k-1} \binom{n}{i} p^i (1-p)^{n-i}$, with $f(i) = 1 - \alpha$, we have
\begin{align*}
    \frac{d}{dp} \left((1 - \alpha) \sum_{i=0}^{k-1} \binom{n}{i} p^i (1-p)^{n-i}\right) = - (n-k+1) (1 - \alpha) \binom{n}{k - 1} p^{k - 1} (1 - p)^{n-k}.
\end{align*}

Finally, for $(1 - \alpha)\sum_{i=k}^{2k-1} \binom{n}{i} p^i (1-p)^{n-i} 2^{-i}\sum_{j=0}^{k-1} \binom{i}{j}$, with $f(i) = (1 - \alpha)2^{-i}\sum_{j=0}^{k-1} \binom{i}{j}$, we have

\begin{align*}
        &\frac{d}{dp} \left((1 - \alpha)\sum_{i=k}^{2k-1} \binom{n}{i} p^i (1-p)^{n-i} 2^{-i}\sum_{j=0}^{k-1} \binom{i}{j}\right) \\
        &= k(1 - \alpha)2^{-k} \binom{n}{k} p^{k - 1} (1 - p)^{n-k}\sum_{j=0}^{k-1} \binom{k}{j} \\
        &\quad - (n-2k+1) (1 - \alpha)2^{1-2k}  \binom{n}{2k - 1} p^{2k - 1} (1 - p)^{n-2k}\sum_{j=0}^{k-1} \binom{2k - 1}{j} \\
        &\quad + \sum_{i=k}^{2k - 2} 2^{-i-1}(1 - \alpha) \left(\sum_{j=0}^{k-1} \binom{i+1}{j} - 2 \sum_{j=0}^{k-1} \binom{i}{j}\right) \binom{n}{i+1} (i+1) p^{i} (1 - p)^{n-i-1}
\end{align*}
where note that $\sum_{j=0}^{k-1} \binom{k}{j} = 2^k - 1$, $\sum_{j=0}^{k-1} \binom{2k - 1}{j} = 2^{2k - 2}$, and
\[\sum_{j=0}^{k-1} \binom{i+1}{j} - 2 \sum_{j=0}^{k-1} \binom{i}{j}
    = \binom{i+1}{0} - \binom{i}{0} - \binom{j}{k - 1} + \sum_{j=0}^{k- 2} \left(\binom{i+1}{j+1} - \binom{i}{j} - \binom{i}{j + 1}\right) = -\binom{i}{k - 1}
\]
so the above quantity becomes
\begin{align*}
        &k(1 - \alpha)2^{-k} \binom{n}{k} p^{k - 1} (1 - p)^{n-k}(2^k - 1)- \frac{1}{2} (n-2k+1) (1 - \alpha)  \binom{n}{2k - 1} p^{2k - 1} (1 - p)^{n-2k}\\
        &\quad - \sum_{i=k}^{2k - 2} 2^{-i-1}(1 - \alpha) \binom{i}{k - 1} \binom{n}{i+1} (i+1) p^{i} (1 - p)^{n-i-1}.
\end{align*}

Putting everything together, we have
\begin{align*}
    & \frac{d}{dp} f(\alpha, p) \\
    &= \frac{\alpha}{2k} n (1 - p)^{n-1} - (n-2k+1) \frac{\alpha(2k-1)}{2k} \binom{n}{2k - 1} p^{2k - 1} (1 - p)^{n-2k} \\
        &\quad + 2k \alpha \binom{n}{2k} p^{2k - 1} (1 - p)^{n-2k} - (n-k+1) (1 - \alpha) \binom{n}{k - 1} p^{k - 1} (1 - p)^{n-k} \\
        &\quad + k(1 - \alpha)2^{-k} \binom{n}{k} p^{k - 1} (1 - p)^{n-k}(2^k - 1)- \frac{1}{2} (n-2k+1) (1 - \alpha)  \binom{n}{2k - 1} p^{2k - 1} (1 - p)^{n-2k}\\
        &\quad + \sum_{i=1}^{2k - 2} \frac{\alpha}{2k} \binom{n}{i+1} (i+1) p^{i} (1 - p)^{n-i-1} \\
        &\quad - \sum_{i=k}^{2k - 2} 2^{-i-1}(1 - \alpha) \binom{i}{k - 1} \binom{n}{i+1} (i+1) p^{i} (1 - p)^{n-i-1}
\end{align*}

Let's group the free terms outside the sum by exponents. First, we note that the coefficients for $p^{2k - 1} (1 - p)^{n-2k}$ is
\begin{align*}
&- (n-2k+1) \frac{\alpha(2k-1)}{2k} \binom{n}{2k - 1} + 2k \alpha \binom{n}{2k} - \frac{1}{2} (n-2k+1) (1 - \alpha)  \binom{n}{2k - 1} \\
&= 2k \binom{n}{2k} \cdot \frac{\alpha k + \alpha - k}{2k} \\
&= \binom{n}{2k}(\alpha k + \alpha - k)
\end{align*}
as $(n-2k+1)\binom{n}{2k - 1} = 2k \binom{n}{2k}$.

The coefficients for $p^{k-1}(1-p)^{n-k}$ is
$
    - (n-k+1) (1 - \alpha) \binom{n}{k - 1} + k(1 - \alpha)2^{-k} \binom{n}{k}(2^k - 1) =-(1-\alpha)2^{-k} k \binom{n}{k}
$
as $(n-k+1)\binom{n}{k-1} = k \binom{n}{k}$.

Therefore, we have
\begin{align*}
    & \frac{d}{dp} f(\alpha, p) \\
    &= \frac{\alpha}{2k} n (1 - p)^{n-1} + \binom{n}{2k}(\alpha k + \alpha - k) p^{2k - 1} (1 - p)^{n-2k} - (1-\alpha)2^{-k} k \binom{n}{k} p^{k-1}(1-p)^{n-k}\\
    &\quad + \sum_{i=1}^{2k - 2} \frac{\alpha}{2k} \binom{n}{i+1} (i+1) p^{i} (1 - p)^{n-i-1} \\
    &\quad - \sum_{i=k}^{2k - 2} 2^{-i-1} (1 - \alpha)\binom{i}{k - 1} \binom{n}{i+1} (i+1) p^{i} (1 - p)^{n-i-1} \\
    &=-(1 - \alpha)k\binom{n}{2k} p^{2k - 1} (1 - p)^{n-2k} \\
    &\quad + \sum_{i=0}^{2k - 1} \frac{\alpha}{2k} \binom{n}{i+1} (i+1) p^{i} (1 - p)^{n-i-1} - \sum_{i=k-1}^{2k - 2} 2^{-i-1} (1 - \alpha)\binom{i}{k - 1} \binom{n}{i+1} (i+1) p^{i} (1 - p)^{n-i-1} \\
    &= -(1 - \alpha)k\binom{n}{2k} p^{2k - 1} (1 - p)^{n-2k} \\
    &\quad + n\left( \sum_{i=0}^{2k - 1} \frac{\alpha}{2k} \binom{n - 1}{i} p^{i} (1 - p)^{n-i-1}  - \sum_{i=k-1}^{2k - 2} 2^{-i-1} (1 - \alpha)\binom{i}{k - 1} \binom{n - 1}{i} p^{i} (1 - p)^{n-i-1}\right)
\end{align*}
where in the second equality, we absorb some terms into the sums, while on the third equality, we use that $\binom{n}{i + 1}(i + 1) = \binom{n - 1}{i} n$.
\end{proof}

\begin{proof}[Proof of~\Cref{lem:df2/dp2}]
Recall that
\begin{align*}
\frac{d}{dp} f(\alpha, p) &= -(1 - \alpha)k\binom{n}{2k} p^{2k - 1} (1 - p)^{n-2k} \\
&\quad + n\left( \sum_{i=0}^{2k - 1} \frac{\alpha}{2k} \binom{n - 1}{i} p^{i} (1 - p)^{n-i-1}  - \sum_{i=k-1}^{2k - 2} 2^{-i-1} (1 - \alpha)\binom{i}{k - 1} \binom{n - 1}{i} p^{i} (1 - p)^{n-i-1}\right)
\end{align*}

We once again apply~\Cref{lem:derivative-of-sum-of-binomial} to each of the sum terms.

For $\sum_{i=0}^{2k - 1} \frac{\alpha}{2k} \binom{n - 1}{i} p^{i} (1 - p)^{n-i-1}$, with $f(i) = \frac{\alpha}{2k}$, we have
\begin{align*}
    \frac{d}{dp}\left(\sum_{i=0}^{2k - 1} \frac{\alpha}{2k} \binom{n - 1}{i} p^{i} (1 - p)^{n-i-1}\right) = -(n-2k) \frac{\alpha}{2k} \binom{n-1}{2k-1} p^{2k-1} (1 - p)^{n-2k-1}.
\end{align*}

For $\sum_{i=k-1}^{2k - 2} 2^{-i-1} (1 - \alpha)\binom{i}{k - 1} \binom{n - 1}{i} p^{i} (1 - p)^{n-i-1}$, with $f(i) = 2^{-i-1} (1 - \alpha) \binom{i}{k - 1}$, we have
\begin{align*}
    &\frac{d}{dp}\left(\sum_{i=k-1}^{2k - 2} 2^{-i-1} (1 - \alpha)\binom{i}{k - 1} \binom{n - 1}{i} p^{i} (1 - p)^{n-i-1}\right) \\
    &= (k-1) 2^{-k} (1 - \alpha) \binom{n-1}{k-1} p^{k-2} (1 - p)^{n-k} \\
    & \quad - (n-2k+1) 2^{1-2k} (1 - \alpha) \binom{2k-2}{k - 1} \binom{n-1}{2k-2} p^{2k-2} (1 - p)^{n-2k} \\
    &\quad + \sum_{i=k-1}^{2k-3} \left(2^{-i-2} (1 - \alpha) \binom{i+1}{k - 1} - 2^{-i-1} (1 - \alpha) \binom{i}{k - 1}\right) \binom{n-1}{i+1} (i+1) p^{i} (1 - p)^{n-i-2} \\
    &= - (n-2k+1) 2^{1-2k} (1 - \alpha) \binom{2k-2}{k - 1} \binom{n-1}{2k-2} p^{2k-2} (1 - p)^{n-2k} \\
    &\quad + \sum_{i=k-2}^{2k-3} 2^{-i-2} (1 - \alpha) \binom{i+1}{k-1} (2k - 3 - i) \binom{n-1}{i+1} p^{i} (1 - p)^{n-i-2} \\
\end{align*}
where on the last equality, we used the identity $\binom{i+1}{k-1} - 2 \binom{i}{k-1} = \binom{i}{k-2} - \binom{i}{k-1} = \binom{i+1}{k-1} \frac{2k - 3 - i}{i + 1}$, and we absorbed the first term into the sum as $i = k - 2$.

Putting everything together, we have
\begin{align*}
&\frac{d^2}{dp^2} f(\alpha, p) \\
&= -(1 - \alpha)k\binom{n}{2k} p^{2k - 2} (1 - p)^{n-2k-1}((2k-1)(1-p)-(n-2k)p) \\
&\quad -n(n-2k) \frac{\alpha}{2k} \binom{n-1}{2k-1} p^{2k-1} (1 - p)^{n-2k-1} \\
& \quad +n(n-2k+1) 2^{1-2k} (1 - \alpha) \binom{2k-2}{k - 1} \binom{n-1}{2k-2} p^{2k-2} (1 - p)^{n-2k} \\
&\quad - n\sum_{i=k-2}^{2k-3} 2^{-i-2} (1 - \alpha) \binom{i+1}{k-1} (2k - 3 - i) \binom{n-1}{i+1} p^{i} (1 - p)^{n-i-2}.
\end{align*}
\end{proof}

\end{document}